\documentclass[12pt]{article}
\usepackage[utf8]{inputenc}

\usepackage{authblk}
\usepackage{amsfonts}
\usepackage{amsmath, amssymb}
\usepackage{amsthm}
\usepackage{graphicx}
\usepackage{url}

\newcommand{\eps}{\varepsilon}
\newcommand{\Z}{{\mathbb Z}}

\newcommand{\R}{{\mathbb R}}

\let\phi=\varphi

\newcommand{\hp}{{\hat p}}
\newcommand{\hq}{{\hat q}}
\newcommand{\hv}{{\hat v}}
\newcommand{\hX}{{\widehat X}}
\newcommand{\hV}{{\widehat V}}

\newcommand{\tX}{{\widetilde X}}
\newcommand{\tV}{{\widetilde V}}

\newcommand{\tp}{{\tilde p}}
\newcommand{\tq}{{\tilde q}}
\newcommand{\tv}{{\tilde v}}

\newcommand{\IP}{{\mathbb P}}
\newcommand{\IE}{{\mathbb E}}

\allowdisplaybreaks

\newtheorem{theo}{Theorem}[section]
\newtheorem{lem}[theo]{Lemma}

\newtheorem{prop}[theo]{Proposition}
\newtheorem{cor}[theo]{Corollary}
\newtheorem{rem}[theo]{Remark}

\newtheorem{claim}[theo]{Claim}

\title{Voting-based probabilistic consensuses and
their applications in distributed ledgers}
\author[1,3]{Serguei Popov}
\affil[1]{\footnotesize{Centro de Matem\'atica, University of
Porto, Porto, Portugal, serguei.popov@fc.up.pt}}

\author[2,3]{Sebastian M\"uller}
\affil[2]{\footnotesize{Aix Marseille Universit\'e, CNRS, Centrale Marseille, I2M - UMR 7373, 13453 Marseille, France, sebastian.muller@univ-amu.fr}}
\affil[3]{IOTA Foundation,
10405 Berlin, Germany}
		
\begin{document}

\maketitle

\begin{abstract}
We review probabilistic models known as majority dynamics 
(also known as threshold voter models) and discuss their possible applications for achieving consensus in cryptocurrency systems. 
In particular,  we show that using this approach 
in a straightforward way
for practical consensus in a Byzantine setting can be problematic and requires extensive further research.
We then discuss the Fast Probabilistic Consensus (FPC) protocol~\cite{fpc},
which circumvents the problems mentioned above
by using external randomness.
\end{abstract}

\section{Introduction}
\label{s_intro}
The interest of academia and industry in Distributed Ledger Technology (DLT) has been steadily increasing in recent years. 
While DLTs are primarily known for their applications within the financial sector, they can become the key enabler for various applications in a wide variety of industries, including smart contracts, digital identity, open APIs, smart healthcare, and data marketplaces.

At the heart of each DLT lies a consensus protocol. 
A consensus protocol enables participants
to reach a consensus 
about which transactions or data are accepted in the ledger. This consensus can be as simple as a binary decision as to whether a cryptocurrency transaction is valid or not. The main challenge of consensus protocols is that they have to be robust against faults and malicious participants who want to corrupt the ledger or delay the system's consensus finding. 

There are many different kinds of consensus protocols, and an ``optimal'' choice depends on the network situation and the actual use case; we refer to, e.g.,~\cite{vademecum}. For instance, one distinguishes between permissioned and permissionless networks, different network sizes, different grades of centrality, and different kinds of  finality. In this paper, we focus on the scenario of large permissionless and decentralized networks where communication costs are essential. The obtained consensus is probabilistic in the sense that consensus is achieved with very high probability. This kind of probabilistic consensus is typical in permissionless networks since  {\it a priori} not all participants are known and a deterministic finality is therefore not achievable. 

The non-trivial question of how participants in a permissionless and decentralized network can reach a consensus at all was answered by the invention of Proof of Work (PoW) or Nakamoto consensus~\cite{nakamoto}. PoW or the ability (computing power) to do a cryptographic puzzle replaces a (physical) identity. In~\cite{nakamoto}, Nakamoto summarizes the protocol as ``(Nodes) vote with their CPU power, expressing their acceptance of valid blocks by working on extending them and rejecting invalid blocks by refusing to work on them.'' Accepted blocks accumulate to chains, and participants agree on the ``longest chain rule''; the distributed ledger equals the longest chain. 

Unfortunately, PoW has many well-known disadvantages, such as high energy consumption and non-scalability.  
For this reason, many efforts have been made to either apply existing consensus protocols to DLTs or develop new ones. 
We again refer to~\cite{vademecum} for an overview of the current consensus protocol landscape.





There 
 has been a great deal of
``classical''
research\footnote{From the side of theoretical computer
science.} on (probabilistic) 
Byzantine consensus protocols, 
see e.g.,\  
\cite{aguilera2012correctness,ben1983another,bracha1987asynchronous,
feldman1989optimal,friedman2005simple,rabin1983randomized}.
However, the disadvantage of 
the approach shared by these papers is 
that they typically require the 
nodes to exchange $O(n^2)$ messages in each round
(where~$n$ is the number of nodes).
This can be a major barrier in situations where communication 
complexity matters
(think about a large number 
of geographically spread nodes, a situation typical in 
cryptocurrency applications). 
There are cryptocurrency systems that use
this approach in practice 
(for example, EOS, Lisk, Neo)
but in order to mitigate the communication complexity issues,
these effectively rely on reduced validator
sets 
(for example, 21 nodes in EOS).

In this paper, we review the probabilistic models known as 
(threshold) voter models or majority dynamics.
Their distinguishing feature compared with the aforementioned probabilistic
Byzantine consensus protocols is their reduced communication 
complexity --- in each round, a node only asks a constant 
number of other nodes for opinions.
This feature makes it very tempting to use them to achieving
a consensus in cryptocurrency systems --- as observed
above, low communication complexity can be very important in
this context. We will, however, see in Section~\ref{s_maj_Byz}
that using them for 
consensus in a secure way is far from being easy.
These models have been extensively studied 
in the probability and statistical physics communities since the 1970s, 
but are perhaps less known in the field of theoretical computer science.
One of the main purposes of this paper is to introduce
as rigorously as possible the (often
highly non-trivial) mathematics of threshold voter models and, in particular, highlight
some possible problems and challenges in their applications
to consensus in Byzantine setting.

\subsection{Probabilistic Voter Models: description and history}
\label{s_history}
There has been extensive research on probabilistic models where, in each round, a node only contacts a small number of other nodes in order to learn their opinions, and possibly change its own. These types of models are usually referred to as \emph{voter models}, and they were introduced in the 1970s by Holley and Liggett~\cite{holley1975ergodic} and Clifford and Sudbury~\cite{CliffSud}; they are part of a larger framework
of \emph{interacting particle systems}~\cite{Liggett}.
Since we are talking about a broad class of models, we cannot 
rigorously define them all at once, but one can \emph{describe}
this class of models in the following way:
\begin{itemize}
\item There is a (finite or infinite) set of nodes
that are identified with the vertices of a connected
graph (usually non-oriented). 
An important case is that of a complete graph where every two nodes are neighbors.
\item At each moment of (discrete or continuous) time,
each node has an \emph{opinion} (also frequently
called \emph{spin} in statistical physics 
literature) $0$ or $1$.
\item At (random or deterministic) moments of time,
a particular node contacts a random subset
of its neighbors and asks them for their current opinions.
It will then update its own opinion according to some
specific rule,
which depends on those queried opinions and also possibly on its own current opinion.
These updates can be synchronous or asynchronous
--- we comment on that in Section~\ref{s_simple_model}.\footnote{It is worth noting at this point that our discussion will sometimes switch from one particular model to another --- after all, our intention is to consider a broad class of models and understand the related phenomena and challenges.}
\item This rule has to be \emph{consistent}:
if a node applies the rule to an all-$0$ or all-$1$
set of opinions, then it will decide on the same opinion.
\item This rule has to be \emph{monotonic}:
 if a node decides on opinion~$i\in\{0,1\}$
according to it, then it would also decide on it
if we flip some of the opinions it received from~$1-i$
to~$i$ (i.e., when the opinion gets more support, a node
cannot stop preferring it).
\end{itemize}
In practical applications, it is also frequently 
necessary
to consider \emph{finalization rules} (i.e., when a node decides
that an opinion~$i$ is final and will not change it anymore)
that may depend on the history, not just on the last opinions 
received. We will comment more on that later; 
for now, let us only observe that it is important 
to understand the behavior of Markovian 
dynamics (i.e., when
the future evolution depends only on the current state of the system,
not on the past) before 
starting to consider
additional rules of that sort.

In general, mathematicians and physicists love voter models because of their rich and interesting limiting behavior
(that can e.g.,\ give rise to hydrodynamic equations, etc.);
a few notable papers are \cite{BraCoxGal01,CoxDurPer00,CoxGri83,PreSpo83,Sow99}.
In the present paper, however, we will 
 concentrate on their consensus applications.
A very important observation is that, in most cases, voter models have only two extremal invariant measures:
one concentrated on the all-$0$ configuration, and the other on all-$1$ configuration\footnote{Basically, 
this is because these configurations are absorbing states
due to the consistency property of the decision rule
and the fact that the evolution of the system is not conservative.}
--- we can naturally call these two configurations ``consensus states''. 
This fact makes voter models interesting for decentralized
consensus applications.

In this class of models, one can distinguish between different kinds of network topologies. 
A first observation is that the better the connectivity of this network, the faster the process converges to one of the extremal measures, 
see e.g., \cite{AbMo:15, Beetal16,  GaZe:18}. 

The next observation that can be made is that the convergence property depends on the distribution of the initial opinions. In cases where the initial densities are chosen at random 
and the density of $1$s 
is significantly different from $\tfrac{1}{2}$, recent research
shows that consensus can quickly be found with a very high probability,
see e.g., \cite{cooper2014power, cooper2015fast, CruiseGanesh14, elsasser2016rapid, fanti2019communication, MoNeTa:13}.
More precisely, the overall communication complexity is $O(n\ln n)$; hence, every node has to issue only $O(\ln n)$ queries.
These results clearly serve as a
motivation for the further study of possible
applications of voting-based protocols for 
consensus in real-world systems.

However, convergence can fail for some fixed initial configurations. Moreover, 
for instance~\cite{TrVuPowerOfFew} contains results that
show the fragility of the protocol on initial configurations. 
There are only a few results in the presence of an adversary. For instance,  in \cite{becchetti2016stabilizing, DoerrStabilizing} robustness was proven for particular cases for an adversary controlling up to $o(\sqrt{n})$ nodes.

Besides the more theoretical works cited above, there have been various more applied works to understand voting protocols. Again, it was shown that these voting protocols may achieve good performances in noiseless and undisturbed networks. However, their performance significantly decreases with 
noise~\cite{GaKuLe:78, GoMaMaBe:15, KaMo:07} or errors~\cite{MoDiAm:04} and may completely fail in a Byzantine setting~\cite{fpcsim}.
This weak robustness in the face of faulty nodes may explain why simple majority voting has not been thoroughly investigated (in the context of practical applications) until recently. 
Let us also cite~\cite{Jedr19,Redner19} for a statistical-physics
approach to applications of the voter model in social sciences and refer to \cite{BistablePotential} for an analysis of opinions dynamics with the help of potentials. 
These works also confirm that the road to consensus might not be
straightforward: as stated in~\cite{Redner19},
``A basic message from these modeling efforts is that incorporating 
any realistic feature of decision making typically leads to either
a dramatically hindered approach to consensus or to the prevention
of consensus altogether.''

We also have to mention another related class of models: cellular automata~\cite{codd1968,WolframCA}. Many kinds of cellular automata (especially when used in consensus applications, e.g.,
\cite{AbMo:15,Beetal16,GaZe:18,TaKiFu:96}) work in a way similar to the models discussed earlier in this section: each node asks its neighbors for opinions and then updates its own opinion using some kind of \emph{deterministic} rule (e.g., a majority one). The dynamics are also usually synchronous (all nodes update their status at the same moment), which makes the evolution of the process completely deterministic (although there may be some randomness in the choice of initial configuration and/or in the choice of the underlying connection graph). This determinism makes the mathematical treatment of these models quite different (and usually also quite harder!) from what will be seen in this paper; for this reason, we have chosen to exclude this (very interesting) class of models from the present analysis and concentrate on models with \emph{probabilistic} dynamics.
Still, we should mention that usage of a deterministic majority dynamics cellular automata was recently proposed in the NKN cryptocurrency system~\cite{nkn}.
On the other hand, in \cite{Ava18,Ava19} a probabilistic voting-based consensus protocol, similar to the models discussed in this paper, was proposed to be used in the Avalanche cryptocurrency system. However, the authors of~\cite{Ava18} and~\cite{Ava19} fail to properly analyze their proposed protocol and the question of whether it has the desired properties remains unclear.
We refer to the end of Section~\ref{s_curse_met} for a more detailed discussion.


\subsection{Outline}
In the following lines, we describe the contents of this paper.
In Section~\ref{s_maj_Byz} we use an example of asynchronous
majority dynamics to explain the methods used for analysing it 
in the simplest setting (Section~\ref{s_simple_model})
and then in the presence of Byzantine actors (Section~\ref{s_enter_Byz}).  To this end, we introduce and explain the concept of ``random walk on a potential landscape'',
which gives an excellent heuristic to understand why, once there is a considerable majority, the protocol quickly converges, but also that there is an intrinsic problem that hinders good performance in Byzantine infrastructures. 

In Section~\ref{s_curse_met} we discuss the 
challenges and difficulties that, due to the metastability 
phenomenon, arise when one intends to 
justify the usage of the majority dynamics models 
for practical consensus in a Byzantine setting.
In particular,  Byzantine actors can cause metastable states;
metastability obliterates
liveness and can compromise safety too.

Then, in Section~\ref{s_fpc} we discuss the Fast Probabilistic Consensus (FPC) protocol:
we introduce the necessary notation and define it in
Section~\ref{s_notation_fpc},
we discuss a few possible adversarial strategies in
Section~\ref{sec:maliciousNodes},
some relevant theoretical results are stated in
Section~\ref{s_theory_fpc}, and
Section~\ref{s_numerical_fpc} is devoted to discussions
of simulation results and generalizations.

\section{Majority dynamics in the presence of Byzantine actors}
\label{s_maj_Byz}
At this point, we would like to mention
that studying these sort of models requires advanced probabilistic tools and their rigorous treatment is in many aspects beyond the scope of this paper. 
Therefore, in this section, we will do our best to intuitively explain the relevant ideas. To this end, we consider a toy model that allows discussing the nature of majority dynamics without getting lost in technicalities.

Another important remark is that, when one is studying the possible applications
 of these models to consensus in DLT,
 it is essential
to be able to obtain \emph{quantitative} estimates for probabilities of relevant events,
and this complicates matters even more.
Elaborating on the last point, it is not infrequent to find in the literature roughly the following definitions
of the safety of the system:
 ``for any $\eps>0$ it is possible to
 adjust parameters in such a way
that the probability that some
 two nodes would decide differently on the validity
of a particular transaction does not exceed~$\eps$''.
This, however, is hardly acceptable, at least in the 
case when one wants to show that the system
scales up well (with respect to the total number
of nodes), for the following reason. 
Typically, the participants in the system want to 
participate in it, which means not only validating transactions
but also issuing them. Then, if those~$n$ nodes
work during time~$T$, there would typically be
$O(nT)$ transactions and \emph{all} of them need to be secure. Hence, for example, a safety estimate with $\eps=\frac{1}{\sqrt{n}}$ will not be of much help (as one will not be
able to use the union bound for estimating the probability
of \emph{at least one} ``safety violation'').
Also, when discussing the liveness guarantees
(i.e., that the consensus on the state of a transaction
eventually occurs),
being able to obtain explicit estimates is also very important. 

To underline the relevance of the last point, a somewhat extreme example is the statement ``Bitcoin addresses are insecure''
with the following ``proof'': indeed, the space of all
possible private keys is finite, therefore
an adversary who just tries them subsequently
would \emph{eventually} find the right one.
However, since nearly every $256$-bit number 
is a valid ECDSA private key, the probability of 
guessing a given private key is of the order $10^{-77}$. Hence, one would need roughly~$10^{77}$
tries to find the right one;  with the current state of technology, this essentially means ``never''.
Therefore, whenever possible, 
we will aim for explicit
estimates (at least those with a clear asymptotic
behavior with respect to the system's size).

\subsection{A simple majority dynamics model and its 
properties}
\label{s_simple_model}
Probably the most accessible way to understand majority dynamics models is by
 \emph{visualizing} them as a random motion on a potential landscape.
To explain this, we will consider a toy model,
which is a variant of a 
discrete-time asynchronous majority dynamics on a 
complete graph. Generally, mathematicians love
playing with toy models:
they are easier to treat, but already contain the main ideas
necessary to understand more complicated models. Usually, toy models give us insight into phenomena to expect in more complicated models of the same kind.
We will also use this model to present some probabilistic
tools and show how they are usually used  in the analysis
of models of this class.
Our toy model is defined in the following
way. The system consists of~$n$ nodes, and at each 
(discrete) time moment one of them is selected uniformly
at random. 
The selected node then chooses\footnote{In the following, for more clarity, we use the verb ``to select''
for the node selected to update its preference in the current
round, and the verb ``to choose'' for the three 
randomly chosen peers
whose current opinions it will use to make the decision.}
three nodes
independently and 
uniformly at random. In particular, it may choose itself,
and also may choose some nodes more than once; in the last
case the corresponding opinions are counted the same number of times. 
The selected node then adopts the majority opinion of the chosen nodes until it is selected again.
In this subsection there will be no Byzantine nodes; that is, for now
we are assuming that all nodes follow the protocol honestly.

Before proceeding, let us also comment on our choice of model.
We have chosen it to be asynchronous. More precisely,
it is an embedded\footnote{That is, tracked at the time moments 
when nodes ``do something''; those (random) time moments form an increasing sequence,
which allows us to introduce that discrete-time process in a natural way.}
discrete-time process of an asynchronous
continuous-time model where, at each time moment,
at most~$1$ node updates its
preference. Another natural choice 
would be a round-based or synchronous model: at a given round,
\emph{all} nodes  choose their random peers
independently and update 
their preferences based on the responses. For this kind of dynamic,
one can likely use the same intuition, but the rigorous treatment
of the synchronous model becomes more difficult
because it can no longer be represented by a nearest-neighbor
random walk;\footnote{One could try to argue that
it would be possible to recover the nearest-neighbor
property if, during the round, one updates the states of the nodes
according to some pre-determined ordering;
but that would destroy the Markov property.} 
analyzing a non-nearest-neighbor random walk 
in a rigorous way may be a \emph{much} more challenging task. 

In the following, let us additionally assume that~$n\geq 20$ and that~$n$  is divisible by~$4$ --- this will spare us from having
to deal with integer parts at some places, and it is a 
toy model anyway.

We denote by~$X_k$ the number of nodes with opinion~$1$
at time~$k$. Assume~$X_k=m$; the number of $1$-opinions among
three independently-chosen (with possible repetitions)
nodes then has the Binomial$\big(3,\frac{m}{n}\big)$
distribution. 
Then, if one of the~$m$ opinion-$1$ nodes
was selected (which happens with probability $\frac{m}{n}$),
it will switch its opinion to~$0$ 
with probability\footnote{Recall that, if 
$\eta\sim \text{Binomial}(3,h)$, then
$\IP[\eta=0]=(1-h)^3$, $\IP[\eta=1]=3(1-h)^2h$, $\IP[\eta=2]=3(1-h)h^2$, $\IP[\eta=3]=h^3$.}
$\big(1-\frac{m}{n}\big)^3 
 +3\big(1-\frac{m}{n}\big)^2\frac{m}{n}$; likewise,
if a node with current opinion~$0$ was selected 
(which happens with probability $1-\frac{m}{n}$),
it will switch its opinion to~$1$ with probability
$\big(\frac{m}{n}\big)^3
+3\big(1-\frac{m}{n}\big)\big(\frac{m}{n}\big)^2$. 
In particular, and as expected, $0$ and~$n$ are absorbing states:
if $X_{k_0}\in\{0,n\}$ for some $k_0$ then $X_k=X_{k_0}$ for all $k>k_0$.
In other words, the process $(X_k, k\geq 0)$
is a (one-dimensional) random walk on~$\{0,\ldots,n\}$ 
with the following transition probabilities: on $X_k=m\in \{0,\ldots,n-1\}$ we have\footnote{Formally,
this also holds for $m\in\{0,n\}$.}
\begin{equation}
\label{df_Xk}
X_{k+1} = 
  \begin{cases}
   m-1, & \text{with probability } 
  p_m
 ,\\
     m+1, & \text{with probability } 
     q_m 
 ,\\
 m, & \text{with probability } v_m=1-p_m-q_m,
  \end{cases}
\end{equation}
where
\begin{align*}
 p_m &=  \tfrac{m}{n}  \big(\big(1-\tfrac{m}{n}\big)^3 
 +3\big(1-\tfrac{m}{n}\big)^2\tfrac{m}{n}\big)
  = \tfrac{m}{n}\big(1-\tfrac{m}{n}\big)  \big(\big(1-\tfrac{m}{n}\big)^2 
 +3\tfrac{m}{n}\big(1-\tfrac{m}{n}\big)\big),\\   
 q_m & = \big(1-\tfrac{m}{n}\big)
     \big(\big(\tfrac{m}{n}\big)^3 
 +3\big(1-\tfrac{m}{n}\big)\big(\tfrac{m}{n}\big)^2\big)
 = \tfrac{m}{n}\big(1-\tfrac{m}{n}\big)  \big(\big(\tfrac{m}{n}\big)^2 
 +3\tfrac{m}{n}\big(1-\tfrac{m}{n}\big)\big), \\
 v_m &= 1-\tfrac{m}{n}\big(1-\tfrac{m}{n}\big)\big(1
  +4\tfrac{m}{n}\big(1-\tfrac{m}{n}\big)\big).
\end{align*}
 

Before going into further details, let us discuss the following
important intuitive idea. Sometimes it is convenient
to \emph{think} about a one-dimensional random walk
as a random motion of a ``Brownian particle''
over a mountainous landscape. 
The particle
moves randomly, but it \emph{prefers} to move
downhill; hopefully, Figure~\ref{f_potentials}
speaks for itself.
\begin{figure}
\begin{center}
\includegraphics[width=0.99\textwidth]{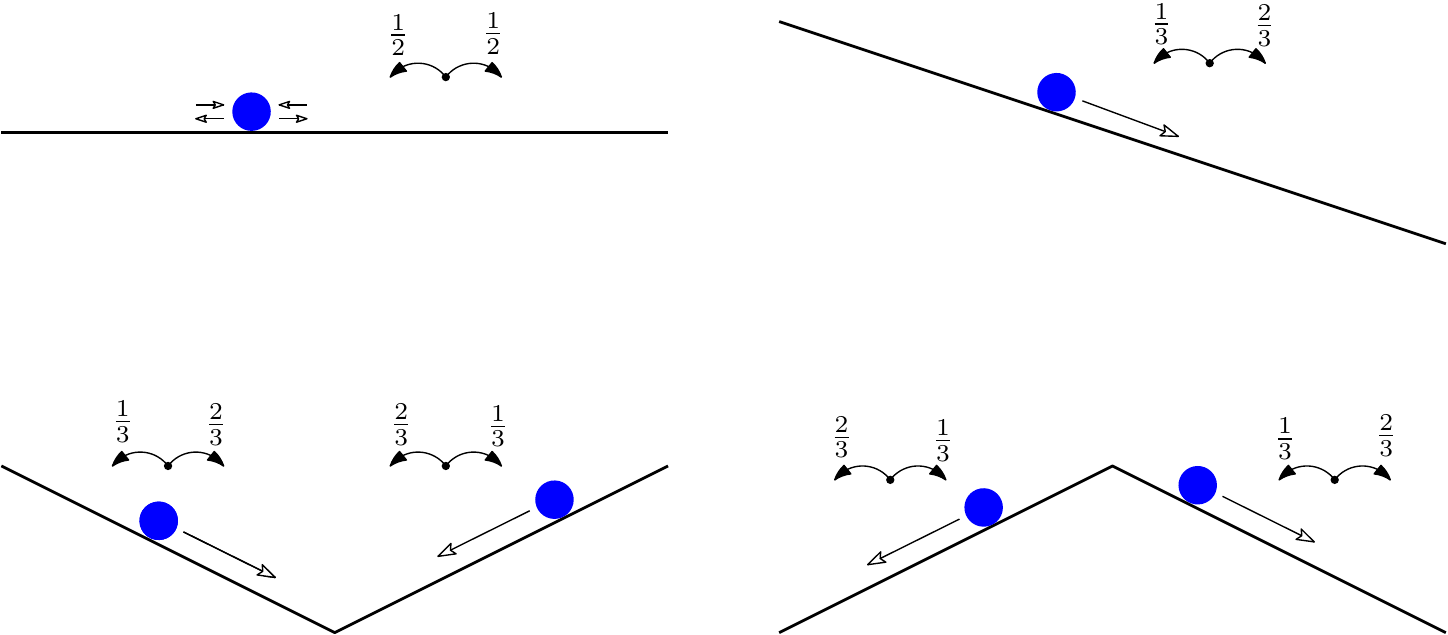}
\end{center} 
 \caption{Intuition: random motion on top of a potential;
the corresponding transition probabilities are indicated
above each of the four pictures.}
\label{f_potentials}
\end{figure}
For example, in the top left picture, the particle has
no \emph{drift} and therefore just moves around in a 
\emph{diffusive} way. 
In the top right picture, 
the particle prefers to go to the right
(i.e., it has a drift in that direction),
on the bottom left picture it would be ``attracted''
by the center of the ``valley'' 
(or \emph{potential well}, as we will mostly call it in the sequel) 
and will probably stay there
for a long time (we will elaborate on that later in this paper),
while in the bottom right picture the particle
(if starting at the center) would randomly
choose one of the slopes and then go downhill. 

Now let us go back to the random walk defined in~\eqref{df_Xk}.
Note that 
the above transition probabilities are symmetric,
in the sense that
$p_{n-m}=q_m$ and $q_{n-m}=p_m$, which implies
that
\begin{equation}
\label{pm(n-m)}
 \frac{p_{n-m}}{q_{n-m}} = \frac{q_m}{p_m}.
\end{equation}

Let us define the function~$V:\{0,\ldots,n-1\}\mapsto \R$ 
(frequently called the \emph{potential})
by $V(0)
=0$ and
\begin{equation}
\label{df_potential}
 V(k) = \sum_{j=1}^k \log \frac{p_j}{q_j}.
\end{equation}
Then, \eqref{pm(n-m)} implies that $V(n-1)=0$
and, in general, $V(m)=V(n-1-m)$
(that is, it is symmetric around $\frac{n-1}{2}$);
in particular, $V(\tfrac{n}{2}-1)=V(\tfrac{n}{2})$. Moreover,
since $p_m>q_m$ for $m<n/2$, we see that $V(m)>0$
for $0<m<n$, it is strictly increasing on~$[0,\tfrac{n}{2}-1]$
and strictly decreasing on~$[\tfrac{n}{2},n-1]$. This potential\footnote{Let us mention that this is only one  possible
definition of the potential; 
for example, in the classical
papers on random walks in random environment, one would rather use the summation $\sum_{j=0}^{k-1}$
in~\eqref{df_potential}. One can note, though, that these definitions lead essentially to similar objects that
are used in similar ways.}
will be the ``landscape profile'' that the walker
walks on, and it in fact resembles the bottom right
profile on Figure~\ref{f_potentials}.
With some more analysis, it is not difficult to
obtain that it closely resembles
the profile on Figure~\ref{f_potential_V}.
\begin{figure}
\begin{center}
 \includegraphics{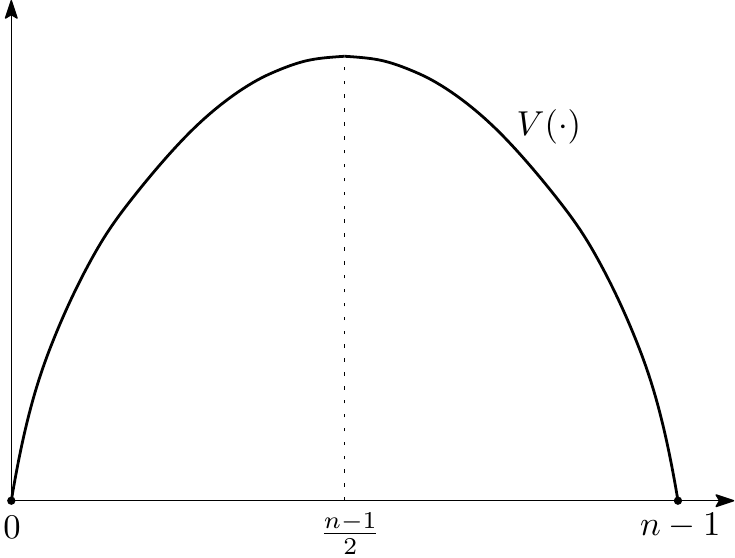}
\end{center}
 \caption{The potential~$V$.}
 \label{f_potential_V}
\end{figure} 
Indeed, we have 
\begin{equation}
\label{formula_f}
\frac{p_m}{q_m} = f(\tfrac{m}{n}),
\qquad \text{where }
f(u) = \frac{(1-u)^2+3u(1-u)}{u^2+3u(1-u)},
\end{equation}
and it is elementary to see that~$f$ is a strictly decreasing
function on the interval~$(0,1)$ with $f(\tfrac{1}{2})=1$
(and also $f(1-u)=1/f(u)$). 
Note also that $\log \frac{p_m}{q_m}\approx 0$ when~$m$
is close to~$\tfrac{n}{2}$.

One of the interesting features of the potential defined above 
is that it can be used for computing the escape 
probabilities from an interval. 
Let $\tau_A=\min\{n\geq 1: X_n \in A\}$ 
denote the hitting time of a set~$A$,
and we abbreviate $\tau_a:=\tau_{\{a\}}$ for 
singletons.
We will also use the notations~$\IP_x$ and~$\IE_x$
for the probability and expectations with respect to the process
started at~$x$.
The next result (see e.g.,\ Lemma~1 of~\cite{Sinai82})
is a useful tool for estimating exit probabilities from intervals.
\begin{lem}
\label{l_prob_exit}
For $0\leq a <x<b\leq n$ it holds that
\begin{equation}
\label{prob_exit}
\IP_x[\tau_b<\tau_a] = \frac{ \sum_{y=a}^{x-1}
  e^{V(y)}}{ \sum_{y=a}^{b-1}  e^{V(y)}}.
\end{equation}
\end{lem}
\begin{proof}
Define
\[
 g(m)=1+\frac{p_1}{q_1}+\cdots
 + \frac{p_1\ldots p_{m-1}}{q_1\ldots q_{m-1}}
  = \sum_{j=0}^{m-1}  e^{V(j)}.
\]
A straightforward computation shows that
$p_mg(m-1)+q_mg(m+1)+v_mg(m)=g(m)$
for $1\leq m\leq n-1$, 
which implies that
the process
$g(X_{j\wedge \tau_{\{0,n\}}})$ is a martingale.\footnote{We
use the notation $a\wedge b := \min\{a,b\}$.}
The Optional Stopping Theorem (see e.g.,\ Section~5.7
of~\cite{Dur10}) then gives that
\[
g(x) = \IP_x[\tau_b<\tau_a] g(b) +
 \big(1-\IP_x[\tau_b<\tau_a]\big)g(a);
\]
solving the above for $\IP_x[\tau_b<\tau_a]$ 
we obtain~\eqref{prob_exit}.
\end{proof}
Let us note that, when dealing with expressions like~\eqref{prob_exit},
it is frequently possible to 
use the heuristic reasoning ``the sum of exponentials is
usually roughly of order of its maximal term''. 
In general, \eqref{prob_exit} justifies the intuition
that the random walker ``prefers to go downhill'' ---
indeed, due to the previous observation, the probability of going
``over a potential wall'' (or ``climbing to the top
of the hill'') should typically be very small.
We also note that, for a fine study of one-dimensional
random walks on a potential, one can
use relevant spectral tools 
similarly to, e.g.,~\cite{ComPop03}.
As an immediate consequence of Lemma~\ref{l_prob_exit}
we obtain the following fact about probabilities
to reach consensus on the initially preferred value
in an already biased situation.
\begin{cor}
\label{c_cons_biased}
 Assume that $0\leq x<\tfrac{n}{2}$. Then
\begin{equation}
\label{eq_cons_biased}
 \IP_x\big[\tau_0<\tau_n\big] 
 = \IP_{n-x}\big[\tau_n<\tau_0\big]
 \geq 1 - x\exp\big(-\big(V\big(\tfrac{n}{2}\big)-V(x)\big)\big).
\end{equation}
In particular, if $\tfrac{x}{n}\leq\alpha<\tfrac{1}{2}$,
it holds that
\begin{equation}
\label{eq_cons_biased_expl}
 \IP_x\big[\tau_0<\tau_n\big] 
 \geq 1 - \tfrac{n}{2}e^{-c_\alpha n}, 
\end{equation}
where $c_\alpha>0$ depends on~$\alpha$ but not on~$n$.
\end{cor}
\begin{proof}
Note that since $\IP_x[\tau_0<\tau_n] 
\geq \IP_x[\tau_0<\tau_{{n}/{2}}]$ for $0<x<\tfrac{n}{2}$, 
we will bound the latter probability.
Then, \eqref{eq_cons_biased}
easily follows from the monotonicity of~$V$
on~$\big[0,\tfrac{n}{2}\big]$
(bound the numerator of~\eqref{prob_exit} from above by the 
maximal term times the number of terms, and bound
the denominator of~\eqref{prob_exit} from below by the 
maximal term).
To obtain~\eqref{eq_cons_biased_expl},
we abbreviate 
$y=\big\lfloor\big(\tfrac{1}{4}
 +\tfrac{\alpha}{2}\big)n\big\rfloor$ and
 write
\[
 V\big(\tfrac{n}{2}\big)-V(x) \geq 
 \sum_{j=x+1}^{y}
 \log\frac{p_j}{q_j}
 \geq \big(\big(\tfrac{1}{4}-\tfrac{\alpha}{2}\big)n-1\big)
   \log\frac{p_{y}}{q_y},
\]
and then use~\eqref{formula_f} and \eqref{eq_cons_biased}. 
\end{proof}
As mentioned before,
the preceding result implies that the ``significantly
leading'' opinion (i.e., an opinion which is shared by at 
least a proportion~$\alpha >\tfrac{1}{2}$) does not
in the end prevail only with an exponentially small (in~$n$)
probability.

Next, we estimate the expected time until
one of the two ``consensus states'' is reached:
\begin{lem}
\label{l_game_duration}
 For any $x\in\{1,\ldots,n-1\}$ it holds that
  $\IE_x\tau_{\{0,n\}}\leq \tfrac{256}{15} n(1+\log n)$.
\end{lem}
\begin{proof}
First, denote 
\[
 Y_t = \begin{cases}
  X_t, & \text{if } X_t\leq \tfrac{n}{2},\\
  n - X_t, & \text{if } X_t> \tfrac{n}{2};
 \end{cases}
\]
by symmetry, the process~$Y$ is  a Markov chain with state space $\big\{0,\ldots,\tfrac{n}{2}\big\}$, and 
it has the same transition probabilities as~$X$ on~$\{0,\ldots,\tfrac{n}{2}-1\}$.
Also, it is clear that it is enough to consider 
the worst-case $x=\tfrac{n}{2}$.

For the process~$Y$,
denote by $T_m$ the expected hitting time of~$0$
starting from~$m$.
It is straightforward to obtain for $0<m< \tfrac{n}{2}$
that 
\begin{equation}
\label{recursion_Tm}
T_m=1+ p_mT_{m-1}+v_m T_m + q_m T_{m+1},
\end{equation}
and 
\begin{equation}
\label{recursion_Tn/2}
 T_{\tfrac{n}{2}}=\tfrac{1}{2}T_{\tfrac{n}{2}}+\tfrac{1}{2}T_{\tfrac{n}{2}-1}+1,
\end{equation}
since $v_{\tfrac{n}{2}}=\tfrac{1}{2}$.
\begin{rem}
Note that~\eqref{recursion_Tm} and \eqref{recursion_Tn/2} imply that
\[
 T_m-T_{m-1} = \frac{1}{p_m}+\frac{q_m}{p_m}(T_{m+1}-T_m)
\]
and 
\[
 T_{\tfrac{n}{2}}-T_{\tfrac{n}{2}-1}=2.
 \]
Using this,
the above system of equations can be solved to
obtain
\begin{equation}
\label{Tm_exact}
 T_m = \sum_{j=1}^m
  \Bigg(2\prod_{\ell=j}^{\tfrac{n}{2}-1}\frac{q_\ell}{p_\ell}
  + \frac{1}{p_j}\sum_{k=j}^{\tfrac{n}{2}-2}
   \prod_{\ell=j}^{k}\frac{q_\ell}{p_{\ell+1}}\Bigg).
\end{equation}
However, the problem with expressions such as~\eqref{Tm_exact}
is that it is not easy to deduce the asymptotic  behavior of  $T_{\tfrac{n}{2}}$ for the $p$'s and $q$'s defined below Equation~\eqref{df_Xk}.\footnote{Even though it is possible to make a reasonable guess about how the terms
in~\eqref{Tm_exact} would behave, writing it all
rigorously is yet another story.}
\end{rem}
To prove Lemma \ref{l_game_duration},
we will therefore take another route, namely the Lyapunov's functions method. This method will
permit us to estimate  $T_{\tfrac{n}{2}}$
without too heavy calculations. Also, it is worth noting
that this method is more robust, in the sense
that it may still work in the non-nearest-neighbor
case when the exact expressions such as~\eqref{Tm_exact}
are not available.

To explain the idea of this method, consider
a function $f:\big\{0,\ldots,\tfrac{n}{2}\big\}\mapsto \R_+$
defined by $f(0)=0$ and $f(m)=T_m$ for $m \in \big\{1,\ldots,\tfrac{n}{2}\big\}$. 
Then, observe that~\eqref{recursion_Tm} and~\eqref{recursion_Tn/2} imply
\begin{equation}
\label{Lyap_-1}
 \IE_m\big(f(Y_1)-f(m)\big) = -1
\end{equation}
for all $m \in \big\{1,\ldots,\tfrac{n}{2}\big\}$.
Now, instead of trying to calculate~$f$, the idea is to  find a  \emph{Lyapunov function}
$g:\big\{0,\ldots,\tfrac{n}{2}\big\}\mapsto \R_+$
which behaves similarly to~$f$ in the sense
 that, for some $\eps>0$
\begin{equation}
\label{Lyap_-C}
 \IE_m\big(g(Y_1)-g(m)\big) \leq -\eps;
\end{equation}
Theorem~2.6.2 of~\cite{MenPopWad17} then implies
 that $\IE_m \tau_0 \leq \eps^{-1}g(m)$
for all~$m>0$.

Let $\delta_n:= \lceil\sqrt{n}\,\rceil
 - \tfrac{n}{\lceil\sqrt{n}\,\rceil}\in [0,1)$
 (observe also that $\tfrac{n}{2}-\lceil\sqrt{n}\,\rceil
 \geq \tfrac{n}{4}$ for $n\geq 20$), and define for $1\leq m \leq \frac{n}{2}$
\[
 \Delta_m =
   \begin{cases}
    \frac{n}{m}+2, & \text{ for }m<\frac{n}{4},\\
    \frac{n}{\tfrac{n}{2}-m}+2, & \text{ for }\frac{n}{4}
          \leq m<\tfrac{n}{2}-\lceil\sqrt{n}\,\rceil ,\\
    \frac{n}{2}-m+2-\delta_n, 
    & \text{ for }\tfrac{n}{2}-\lceil\sqrt{n}\,\rceil
                \leq m \leq \tfrac{n}{2}.
   \end{cases}
\]
Note that there is continuity in the sense
that the first formula is also valid for $m=\tfrac{n}{4}$
and the second formula is also valid\footnote{We
had to introduce~$\delta_n$ for that reason; for following
the subsequent calculations in an easier way,
one can just assume that $\sqrt{n}$ is an integer so $\delta_n=0$.} for
$m=\tfrac{n}{2}-\lceil\sqrt{n}\,\rceil$.
Let us define $g(m)=\Delta_1+\cdots+\Delta_m$.
Then, observe that (recall that $v_{\tfrac{n}{2}}=\tfrac{1}{2}$)
\begin{equation}
\label{drift_n/2}
  \IE_{\tfrac{n}{2}}\big(g(Y_1)-g(\tfrac{n}{2}) \big)
   = -\tfrac{1}{2}\Delta_{\tfrac{n}{2}}\leq -\tfrac{1}{2},
\end{equation}
and, for $m<\tfrac{n}{2}$
\begin{align}
 \IE_m\big(g(Y_1)-g(m)\big) &=
p_mg(m-1)+v_mg(m)+q_mg(m+1)-g(m)\nonumber\\
&= p_m(g(m-1)-g(m))+q_m(g(m+1)-g(m))\nonumber\\
&= -p_m\Delta_m + q_m\Delta_{m+1}.
 \label{g_drift_calc}
\end{align}
Note also that, since $p_m\geq q_m$,
for the sake of upper bounds we can always drop
the ``$+2$'' (as well as ``$+2-\delta_n$'') part
from the calculations.
For $m<\tfrac{n}{4}$ (equivalently,
$\tfrac{m}{n}<\tfrac{1}{4}$), we have
\begin{align*}
 -p_m\Delta_m + q_m\Delta_{m+1}
& \leq -\tfrac{m}{n}\big(1-\tfrac{m}{n}\big) 
 \Big(\tfrac{n}{m}\big(\big(1-\tfrac{m}{n}\big)^2 
 +3\tfrac{m}{n}\big(1-\tfrac{m}{n}\big) \big)\\
&\qquad\qquad\qquad \qquad
-\tfrac{n}{m+1}\big(\big(\tfrac{m}{n}\big)^2 
 +3\tfrac{m}{n}\big(1-\tfrac{m}{n}\big)\big)\Big)\\
 &\leq - \big(1-\tfrac{m}{n}\big) 
  \Big(\big(1-\tfrac{m}{n}\big)^2-\big(\tfrac{m}{n}\big)^2\Big)\\ & = - \big(1-\tfrac{m}{n}\big) \big(1-2\tfrac{m}{n}\big)\\
  &\leq - \tfrac{3}{8}.
\end{align*}
Then, for $\frac{n}{4}\leq m<\tfrac{n}{2}-\lceil\sqrt{n}\,\rceil$
we can write
\begin{align*}
\lefteqn{
 -p_m\Delta_m + q_m\Delta_{m+1}
 } \hphantom{***}\\
& \leq -\tfrac{m}{n}\big(1-\tfrac{m}{n}\big) 
 \Big(\tfrac{n}{\tfrac{n}{2}-m}\big(\big(1-\tfrac{m}{n}\big)^2 
 +3\tfrac{m}{n}\big(1-\tfrac{m}{n}\big) \big)\\
&\qquad\qquad\qquad \qquad
 -\tfrac{n}{\tfrac{n}{2}-m-1}\big(\big(\tfrac{m}{n}\big)^2 
 +3\tfrac{m}{n}\big(1-\tfrac{m}{n}\big)\big)\Big)\\
 &= -\tfrac{m}{n}\big(1-\tfrac{m}{n}\big) 
 \Big(\tfrac{n}{\tfrac{n}{2}-m}\big(1-\tfrac{m}{n}\big)^2 
 - \tfrac{n}{\tfrac{n}{2}-m}\big(\tfrac{m}{n}\big)^2
 + \tfrac{n}{\tfrac{n}{2}-m}\big(\tfrac{m}{n}\big)^2
 \\
  & \qquad 
  - \tfrac{n}{\tfrac{n}{2}-m-1}\big(\tfrac{m}{n}\big)^2
  -3\tfrac{m}{n}\big(1-\tfrac{m}{n}\big)
  \frac{n}{(\tfrac{n}{2}-m)(\tfrac{n}{2}-m-1)}
 \Big)\\
 &=  -\tfrac{m}{n}\big(1-\tfrac{m}{n}\big) 
 \Big( 2 - \big( \big(\tfrac{m}{n}\big)^2 +
 3\tfrac{m}{n}\big(1-\tfrac{m}{n}\big)\big)
  \frac{n}{(\tfrac{n}{2}-m)(\tfrac{n}{2}-m-1)}
  \Big)\\
  \intertext{{\footnotesize \quad (note that
  $h(1-h)\in [\tfrac{3}{16},\tfrac{1}{4}]$ for
  $h\in [\tfrac{1}{4},\tfrac{1}{2}]$ and that 
  the last fraction is $\leq 1$)}}
  &\leq - \tfrac{3}{16}.
\end{align*}
Finally, for $\tfrac{n}{2}-\lceil\sqrt{n}\,\rceil\leq m 
< \tfrac{n}{2}$
we have 
\begin{align*}
 -p_m\Delta_m + q_m\Delta_{m+1}
& \leq -\tfrac{m}{n}\big(1-\tfrac{m}{n}\big) 
 \Big(\big(\tfrac{n}{2}-m\big)\big(\big(1-\tfrac{m}{n}\big)^2 
 +3\tfrac{m}{n}\big(1-\tfrac{m}{n}\big) \big)\\
&\qquad\qquad\qquad \qquad
 -\big(\tfrac{n}{2}-m-1\big)\big(\big(\tfrac{m}{n}\big)^2 
 +3\tfrac{m}{n}\big(1-\tfrac{m}{n}\big)\big)\Big)\\
 & =  -\tfrac{m}{n}\big(1-\tfrac{m}{n}\big) 
 \Big(\big(\tfrac{n}{2}-m\big)\big(1-2\tfrac{m}{n}\big) 
+\big(\tfrac{m}{n}\big)^2 
 +3\tfrac{m}{n}\big(1-\tfrac{m}{n}\big)\Big)\\
  \intertext{{\footnotesize  (again use that
  $h(1-h)\in [\tfrac{3}{16},\tfrac{1}{4}]$ for
  $h\in [\tfrac{1}{4},\tfrac{1}{2}]$ and that 
  the first term in the parentheses is nonnegative)}}
 &\leq - \tfrac{15}{128}.  
\end{align*}
Gathering the pieces, we find that~\eqref{Lyap_-C}
holds with $\eps=\tfrac{15}{128}$.
It is also not difficult to obtain that
\begin{align*}
g\big(\tfrac{n}{2}\big) &=
 \Delta_1+\cdots+\Delta_{\tfrac{n}{2}}\\
 &\leq n + \sum_{m=1}^{\tfrac{n}{4}-1}\frac{n}{m}
  + \sum_{m=\tfrac{n}{4}}^{\tfrac{n}{2}-\lceil\sqrt{n}\,\rceil-1}
  \frac{n}{\tfrac{n}{2}-m}
   + \sum_{m=\tfrac{n}{2}-\lceil\sqrt{n}\,\rceil}^{\tfrac{n}{2}}
   \big(\tfrac{n}{2}-m\big)
   \\
 &  \leq 2n (1+\log n ).
\label{est_g(n/2)}   
\end{align*}
Then, using Theorem~2.6.2 of~\cite{MenPopWad17} 
we conclude the proof of Lemma~\ref{l_game_duration}.
\end{proof}

The constant~$\tfrac{256}{15}$ in Lemma~\ref{l_game_duration}
can
be improved, but we need not concern ourselves with that here.
One due observation, however, is that this result
gives the correct order of the expectation of the 
hitting time (at least for the case $\log x \asymp \log n$);
this can be easily seen from the fact that
$p_m+q_m\asymp \tfrac{m}{n}$ as~$m$ decreases to~$0$.

\begin{cor}
\label{c_game_duration}
 For any~$x$ and any positive integer~$k$ it holds that
\begin{equation}
\label{tail_game_duration}
 \IP_x\Big[\tau_{\{0,n\}}>
  k \big\lceil \tfrac{512}{15}n(1+\log n)\big\rceil\Big] \leq 2^{-k}.
\end{equation}
\end{cor}

\begin{proof}
Indeed,  Chebyshev's inequality together
with Lemma~\ref{l_game_duration} imply that
\[
 \IP_x\big[\tau_{\{0,n\}}>\tfrac{512}{15}n(1+\log n)
 \big] \leq \tfrac{1}{2}.
\]
A straightforward coin-tossing argument 
(think about dividing the time interval 
$\big[0,k \big\lceil \tfrac{512}{15}n(1+\log n)\big\rceil\big]$
into~$k$
subintervals of length 
$\big\lceil \tfrac{512}{15}n(1+\log n)\big\rceil$ and note 
that on each subinterval $\tau_{\{0,n\}}$
occurs with probability at least~$\tfrac{1}{2}$) then 
implies~\eqref{tail_game_duration}.
\end{proof}

Now, let us summarize what we have obtained for our
majority dynamics toy model. We have seen that:
\begin{itemize}
 \item Lemma~\ref{l_game_duration} and
    Corollary~\ref{c_game_duration} imply
    that the system converges  rather quickly on one of the two
    consensus states. Note also
    that, since a particular node is only selected roughly once
    every~$n$ rounds, it will,
    on average, only 
    update its state a logarithmic number of times.
 \item Corollary~\ref{c_cons_biased} shows that,
  if the overall opinions are already biased 
  to one of the sides, then with high probability
  the system will achieve consensus on that prevailing
  opinion.
\end{itemize}
 Also, thanks to the \emph{quantitative} 
  estimates that we have obtained, one can prove that 
  the system can scale well: if it has to come 
  to consensus on $O(n^\gamma)$ number of bits 
  for some fixed~$\gamma>1$, 
  we will still be able to prove that it can do this 
  with high probability simply by using the union
  bounds.
Having in mind practical applications, one can also
analyze local finalization rules for the nodes,
for example of the following type: ``if last $N$ instances
of my queries have all returned the same result,
then consider this opinion to be final''
(as previously mentioned, these local
rules may be needed because in practice there might
be no central authority who is able to 
tell that the system finds itself in one of the two
consensus states).

Due to the above considerations, one may be tempted to think
that a system based on a majority dynamics consensus
 can be used for practical
 applications that demand low communication complexity.
This may indeed be the case in situations when
all nodes are honest.
However, in the next section, we will see that
the presence of Byzantine nodes can significantly complicate 
the situation ---
it will no longer be possible to analyze the system 
in such a simple-and-elementary way 
as in Section~\ref{s_simple_model}.

\subsection{Enter Byzantine nodes}
\label{s_enter_Byz}
Before discussing those, we need to develop 
a better understanding of (nearest-neighbor) random walks on top of 
a potential profile. Let us consider a 
one-dimensional random walk~$\hX$ on the state
space $[a,b]\subset\Z$; analogously to~\eqref{df_Xk}, 
we denote $\hp_x:=\IP_x[\hX_1=x-1]$, 
$\hq_x:=\IP_x[\hX_1=x+1]$, $\hv_x:=\IP_x[\hX_1=x]$,
$\hp_x+\hq_x+\hv_x=1$ for all $x\in[a,b]$
(and, naturally, $\hp_a=\hq_b=0$).
Then, define the potential~$\hV$
analogously to~\eqref{df_potential},
with $\hp$'s and $\hq$'s on the place
of $p$'s and $q$'s.

Now, the question is: suppose that we know
what the potential~$\hV$ looks like; 
for example, it could be as shown in
Figure~\ref{f_metastability_full1}.
\begin{figure}
\begin{center}
 \includegraphics[width=0.73\textwidth]{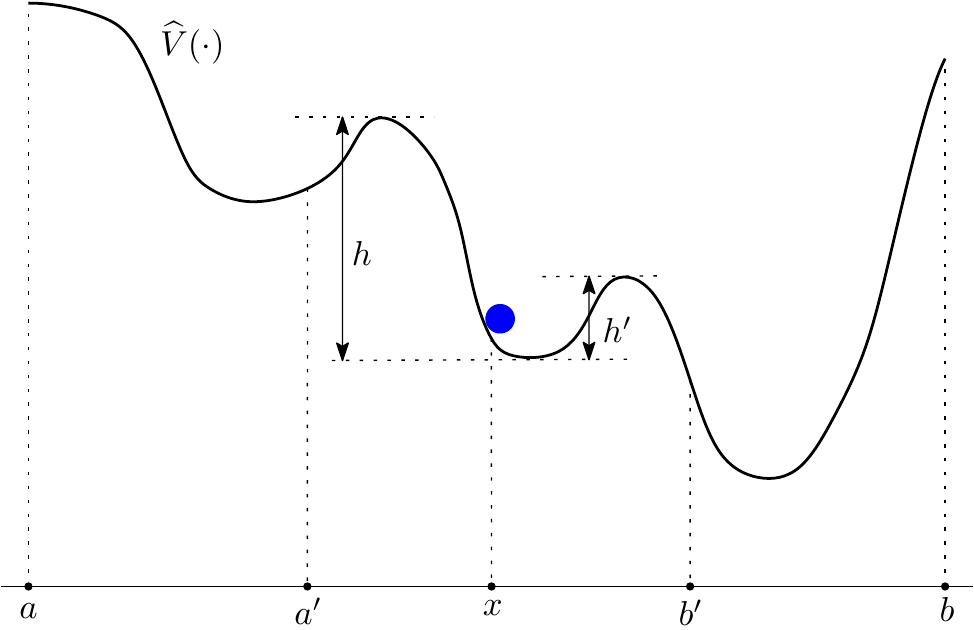}
\end{center}
 \caption{Random walk on a potential.}
\label{f_metastability_full1}
\end{figure} 
Based on this, what can we say about the properties 
of the process~$\hX$?

In the following few lines, we will formulate
some principles that would permit
us to understand the behavior of the random walk as
a function of the potential, at least on an 
intuitive level. 
This intuition can be turned into rigorous arguments; the exact formulations
(let alone proofs) are beyond the scope 
of this paper; but let us cite the seminal paper~\cite{Sinai82}
where this sort of analysis was applied to random walks
in random environments to characterize its highly non-trivial 
behavior in the recurrent regime. 
We will also assume that the random walk does not have any absorbing states (in contrast to the one in the previous section) -- we will see that this will be the case in the presence of Byzantine nodes anyway.  
Let us also keep in mind that the following intuition typically applies to various reasonable definitions of the potential and that we are trying to understand processes with \emph{large} state spaces. For this reason we will not try to be excessively precise when making pictures, e.g., we represent the potential by continuous graphs even though it is defined on a discrete domain,
and, for example, we pretend that $an$ and $an+1$ are indistinguishable on the picture when~$n$ is large.
\begin{claim}
\label{cl_behavior_potentials}
The random walk on top of a potential 
roughly behaves in the following way:
\begin{itemize}
\item[(i)] it ``prefers'' to go downhill;
\item[(ii)] the probability that it goes
in an ``improbable direction'' is roughly\footnote{That is,
up to factors of a smaller order.}
the exponential of minus the difference of ``heights
to overcome'' (for example,
on Figure~\ref{f_metastability_full1}, 
$\IP_x[\tau_{a'}<\tau_{b'}]\approx e^{-(h-h')}$);
\item[(iii)] its stationary distribution at~$x$ 
 is roughly proportional to~$e^{-\hV(x)}$;
so, in the long run, the process will stay for
the overwhelming amount of time at the bottom(s)
(global minima)
of the potential;
\item[(iv)]  The time to go out of a potential well of depth~$h'$
(such as the one where the walker currently is 
on Figure~\ref{f_metastability_full1})
is \emph{roughly} an Exponential random variable
with the expected value~$\approx e^{h'}$.
\end{itemize}
\end{claim}
We justify the above in the following way.
Item (i) is clear from the definition, 
and~(ii) is essentially a consequence of Lemma~\ref{l_prob_exit}
(together with the observation about sums of
exponentials after it).
As for~(iii), from the fact that the Markov chain
is reversible, we obtain that the 
stationary distribution at~$x$ is \emph{exactly}
proportional to 
$\frac{\hq_a\cdots\hq_{x-1}}{\hp_{a+1}\cdots \hp_x}$;
now, this is in fact \emph{almost}\footnote{Up to a, typically bounded, multiplicative factor.} $e^{-\hV(x)}$.
Regarding~(iv), 
note that it is quite common for the hitting
times of Markov chains to have approximately Exponential distribution; see 
e.g.,\ \cite{Ald82,AldBro92,AldBro93,Keil80,MQS21,ManSco19}.
\begin{figure}
 \begin{center}
 \includegraphics{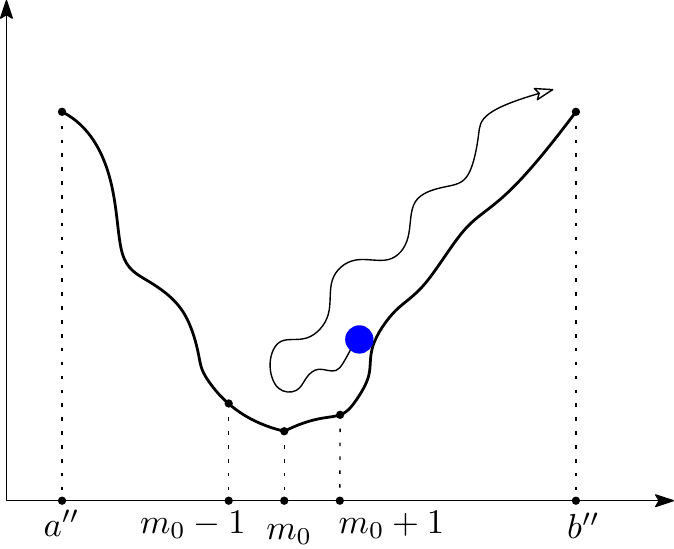}
\end{center}
 \caption{Escaping from a potential well.}
\label{f_pot_well}
\end{figure} 
Let us now explain why the expected time to go out of a potential
well should be of order of exponential of its depth.
Consider the situation depicted in Figure~\ref{f_pot_well}:
$\hV(a'')=\hV(b'')$, $m_0$ is where the minimum of~$\hV$
on the interval~$[a'',b'']$ is attained, and,
for the sake of cleanness, let us
also suppose that
$\hV(m_0 \pm 1)-\hV(m_0)$
are bounded away from~$0$.
\begin{figure}
 \begin{center}
 \includegraphics[width=\textwidth]{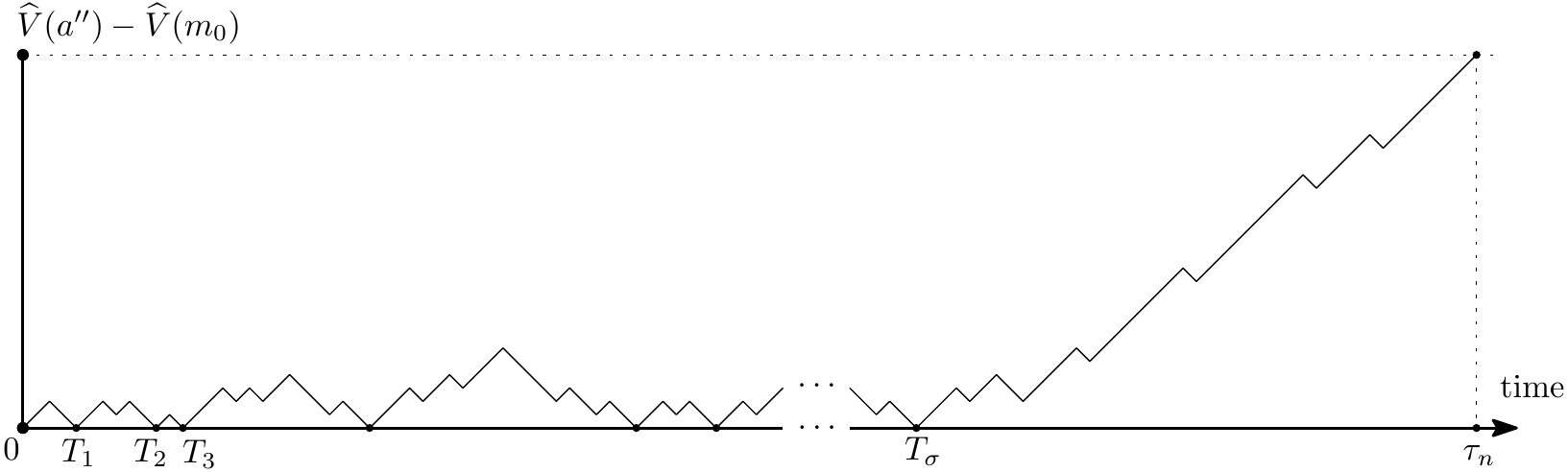}
\end{center}
 \caption{Calculating the distribution of the time 
 to escape from a potential well (the vertical coordinate
 represents the ``height'' of the walk above~$\hV(m_0)$
 with respect to the potential).}
\label{f_hitting_time}
\end{figure} 
 Then, the idea is to consider \emph{excursions}
 of the Markov chain out of $m_0$ before hitting~$\{a'',b''\}$,
 i.e., climbing to the level $\hV(a'')-\hV(m_0)$
potential-wise (it should be noted that, in general, 
analysis of excursions is a very common
 approach when studying random walks).
 Let~$\sigma$ be the number of those
 excursions, and $T_1,T_2,\ldots,T_\sigma$ be the time moments
 when the process is at~$m_0$ 
 (see Figure~\ref{f_hitting_time}).
Using Lemma~\ref{l_prob_exit}, we obtain that,
up to factors of a smaller order,
\[
 \IP_{m_0\pm 1}[\tau_{\{a'',b''\}}<\tau_{m_0}] 
\approx e^{-(\hV(a'')-\hV(m_0))}. 
\]
This implies that the number of excursions~$\sigma$
 has approximately Geometric (and therefore 
 approximately Exponential) distribution with mean
around $e^{(\hV(a'')-\hV(m_0))}$;
therefore, the expected value of~$T_\sigma$ itself
should be of order  $e^{(\hV(a'')-\hV(m_0))}$ as well
(because it is roughly~$\sigma$ times the mean
excursion length).
We also need to argue that the length
of the last excursion (from~$m_0$ to~$\{a'',b''\}$)
is typically negligible compared to~$T_\sigma$.
To see that, note that
it is a trajectory of a \emph{conditioned} 
(on the event $\{\tau_{\{a'',b''\}}<\tau_{m_0}\}$) random
walk, the so-called Doob's $h$-transform of the original one,
see e.g.,\ Section~4.1 of~\cite{Popov21}
and also Exercise~4.58 of~\cite{Ross_Models09}.
Essentially, it will behave as a random walk with
inverted drift (i.e., towards the interval's extremes), so 
it will typically have a length that does not exceed much
in the order of magnitude the length of the interval~$[a'',b'']$;
therefore, it will be typically negligible compared 
to~$e^{(\hV(a'')-\hV(m_0))}$.

Now, with Claim~\ref{cl_behavior_potentials} to hand,
we can explain how the random walk on the above
potential will typically behave. Of course,
we assume that the size of the interval~$[a,b]$
is large.

Assume that the random walk 
on Figure~\ref{f_metastability_full} starts somewhere 
near site~$a$ (that is, ``on top'' of the potential).
It will quickly come to the first \emph{metastable state}
(the leftmost potential well on Figure~\ref{f_metastability_full}), 
and then spend roughly~$e^{h_1}$ time there.
Then, eventually, it will overcome the potential
wall separating the first potential well from the second
one, and then spend around~$e^{h_2}$ time there
(that would be the second metastable state of the process).
Then, finally, it will climb the second wall of height~$h_2$
and come to the deepest 
potential well (the \emph{ground state}). After that,
essentially, it will stay there (of course, very occasionally 
it will go back to the metastable states since
it is a \emph{finite} Markov chain, but the time 
spent there would be negligible in comparison 
to the time spent in the ground state).
\begin{figure}
\begin{center}
 \includegraphics[width=0.73\textwidth]{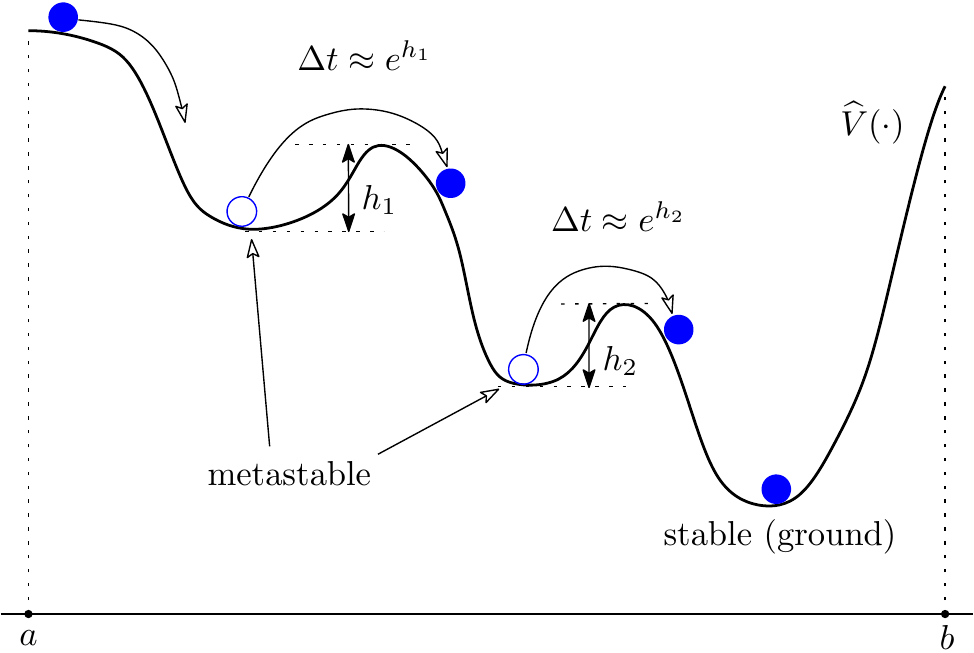}
\end{center}
 \caption{The metastability phenomenon, and time spent
 in metastable states.}
\label{f_metastability_full}
\end{figure}
The term ``metastable state'' thus means a state\footnote{Here,
the word ``state'' is used in a loose sense: it does not mean
a particular point of the state space of the Markov chain,
but rather a region close to a given local minimum of the potential.}
of the system which is not
the one where the system passes the majority of the time in the long
run (which would be around the \emph{global}
 minimum of the potential), but where the system still spends
quite a long time before finding its way out of there. 
 
In general, metastability is a ubiquitous phenomenon
in the context of statistical physics and stochastic 
processes; a complete overview of the relevant
literature would be impossible here, but let us refer to
\cite{BEGK04,BovHol16,CGOV84,FreWen12,KS21,OliVar05,PenLeb71}
for a start.
We mention also that this process of going from one 
metastable state to the next one is sometimes referred to 
as \emph{aging} in the literature,
see e.g.,~\cite{DemGuiZei01,EnrSabZin09}.

Now, we return to the majority dynamics model of 
the previous section, but this time with Byzantine nodes.
Fix a positive number $q\in(0,\tfrac{1}{2})$.
We will now assume that, among those~$n$ nodes,
$(1-q)n$ are \emph{honest} (i.e., they follow the prescribed
protocol), and $qn$ are \emph{Byzantine}.
In principle, the Byzantine nodes can behave in any matter; here
we assume that they all are controlled by a single
entity, called \emph{the adversary}. This entity is able to know the current
state of the system; i.e., the adversary is 
\emph{omniscient}.\footnote{For sure,  complete omniscience is unrealistic;
but, on the other hand, it is unclear how exactly to quantify the degree of knowledge 
that the adversary has about the network, and it is also unclear up 
to ``which point'' (whatever that means) it can gain that knowledge (now or in the future). 
However, it is clear that the adversary can obtain \emph{some} knowledge about the state of the 
network (by directly observing it and also maybe by doing some statistical analysis) 
and will try to do that in the case when the security depends on the network state's obfuscation.
Therefore, it is reasonable to assume omniscience to be on the safe side.}
Let us warn the reader already at this point that 
the adversary can adopt a multitude of different
\emph{strategies} -- these are, roughly, the 
``behaviors'' that the adversary uses to achieve
its goals. These goals may be, for example, to delay the 
consensus among the honest nodes, or to prevent
it from occurring altogether by having different nodes finalize
on different opinions, or to have some nodes
reach their final opinion while keeping the others
undecided.

Let us consider a very simple adversarial strategy:
``Help the weakest''. That is, when queried, the adversarial nodes
will always vote on the less popular opinion
among the honest nodes. 
The goal now is to try\footnote{``Try not! Do, or do not!''
\copyright} 
to analyze the system similarly
to the previous section, via a one-dimensional 
nearest-neighbor random walk.

Denote by~$\tX_k$ be the number of \emph{honest}
nodes with opinion~$1$
at time~$k$. Byzantine nodes have no \emph{real} opinions,
and if a Byzantine node is \emph{selected}, then
the state of the system will not change.
Let us assume~$\tX_k=m$, and consider the following two cases.

\paragraph{Case 1: $m\leq \tfrac{(1-q)n}{2}$.}
In this case, any Byzantine node, if queried, 
will vote for~$1$. 
So, the number of $1$-opinions among
three independently chosen 
nodes will have the Binomial$\big(3,\frac{m}{n}+q\big)$
distribution. 
Then, if one of the~$m$ honest opinion-$1$ nodes
was selected (which happens with probability $\frac{m}{n}$),
it will switch its opinion to~$0$ with probability
$\big(1-\frac{m}{n}-q\big)^3 
 +3\big(1-\frac{m}{n}-q\big)^2\big(\frac{m}{n}+q\big)$;
 likewise,
if an honest node with current opinion~$0$ was selected 
(which happens with probability $1-q-\frac{m}{n}$),
it will switch its opinion to~$1$ with probability
$\big(\frac{m}{n}+q\big)^3
+3\big(1-\frac{m}{n}-q\big)\big(\frac{m}{n}+q\big)^2$. 

\paragraph{Case 2: $m > \tfrac{(1-q)n}{2}$.}
In this case, a Byzantine node, if queried, 
will vote for~$0$, and therefore
the number of $1$-opinions among
three independently chosen 
nodes has the Binomial$\big(3,\frac{m}{n}\big)$
distribution. 
Then, as before, 
if one of the~$m$ honest opinion-$1$ nodes
was selected,
it will switch its opinion to~$0$ with probability
$\big(1-\frac{m}{n}\big)^3 
 +3\big(1-\frac{m}{n}\big)^2\big(\frac{m}{n}\big)$;
if an honest node with current opinion~$0$ was selected, 
it will switch its opinion to~$1$ with probability
$\big(\frac{m}{n}\big)^3
+3\big(1-\frac{m}{n}\big)\big(\frac{m}{n}\big)^2$. 

Note, in particular, that $0$ and~$(1-q)n$ are 
\emph{not} absorbing states anymore ---
because even if all the honest nodes agree,
there is a chance that the selected
honest node will choose at least two Byzantine
ones for opinions and those will convince it to switch.

That is, we find that the process $(\tX_k, k\geq 0)$
is a (one-dimensional) random walk on~$\{0,\ldots,(1-q)n\}$, 
and on $\tX_k=m$ we have
\begin{equation}
\label{df_Xk_Byz}
\tX_{k+1} = 
  \begin{cases}
   m-1, & \text{with probability } 
  \tp_m
 ,\\
     m+1, & \text{with probability } 
     \tq_m 
 ,\\
 m, & \text{with probability } \tv_m=1-\tp_m-\tq_m,
  \end{cases}
\end{equation}
where
\begin{align*}
 \tp_m &= \begin{cases}
   \tfrac{m}{n} \big(\big(1-\tfrac{m}{n}-q\big)^3
 +3\big(\tfrac{m}{n}+q\big)\big(1-\tfrac{m}{n}-q\big)^2\big), &
 \text{ if } m\leq \tfrac{(1-q)n}{2},\\
 \tfrac{m}{n} \big(\big(1-\tfrac{m}{n}\big)^3
 +3\tfrac{m}{n}\big(1-\tfrac{m}{n}\big)^2\big),
 \vphantom{\displaystyle\sum^a}&
 \text{ if } m > \tfrac{(1-q)n}{2},
 \end{cases}
\\   
 \tq_m & = \begin{cases}
   \big(1-\tfrac{m}{n}-q\big) \big(\big(\tfrac{m}{n}+q\big)^3
 +3\big(\tfrac{m}{n}+q\big)^2\big(1-\tfrac{m}{n}-q\big)\big), &
 \text{ if } m\leq \tfrac{(1-q)n}{2},\\
\big(1-\tfrac{m}{n}-q\big) \big(\big(\tfrac{m}{n}\big)^3
 +3\big(\tfrac{m}{n}\big)^2\big(1-\tfrac{m}{n}\big)\big),
 \vphantom{\displaystyle\sum^a}&
 \text{ if } m > \tfrac{(1-q)n}{2}.
 \end{cases}
\end{align*}
We then define the potential~$\tV$ 
analogously to~\eqref{df_potential} (with $\tp$'s
and $\tq$'s on the place of $p$'s and $q$'s).
\begin{figure}
\begin{center}
 \includegraphics{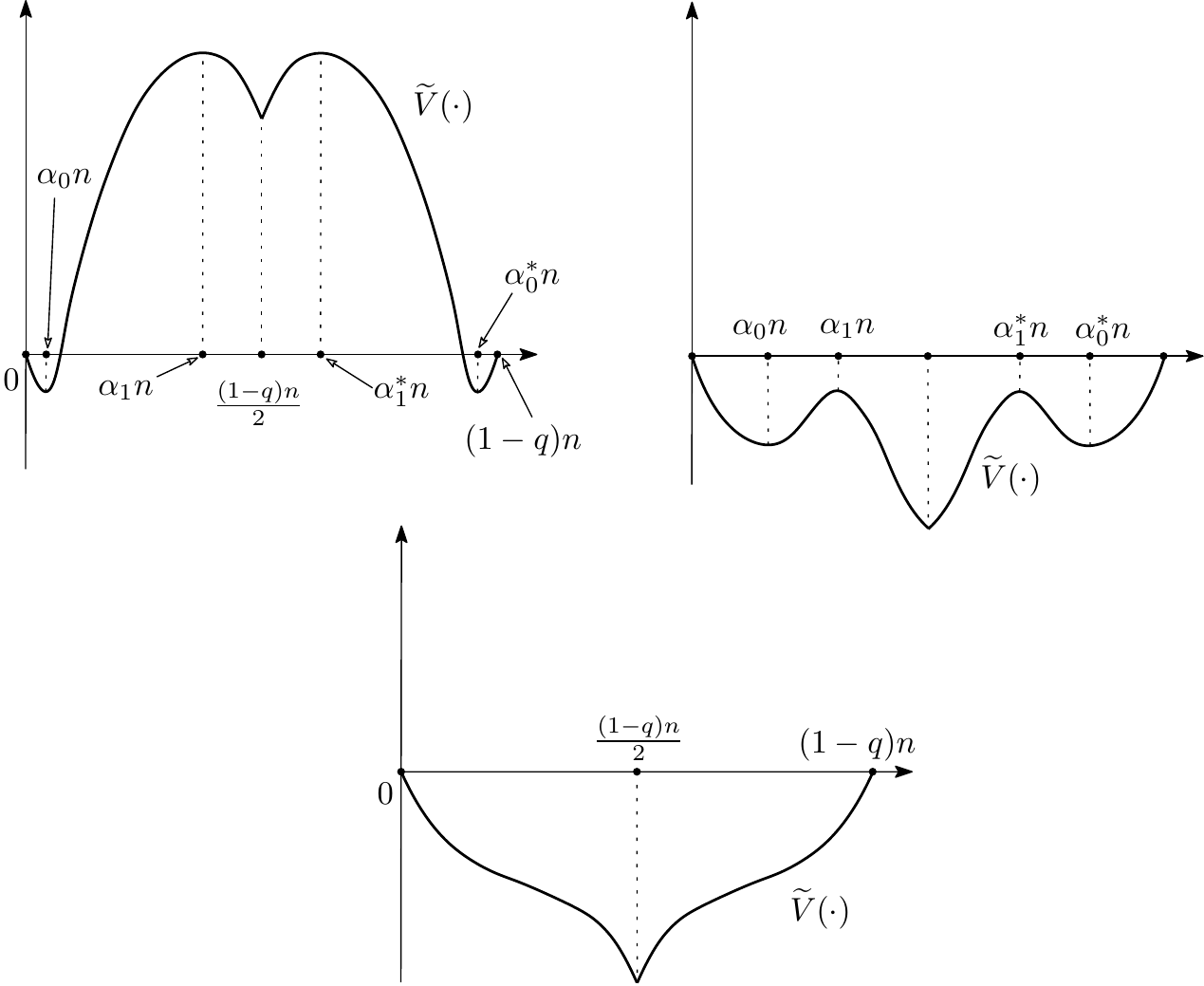}
\end{center}
 \caption{The potential profile for the majority
 dynamics with Byzantine nodes: 
 with $q<0.09029$ (top left), $0.09029<q<\frac{1}{9}$
 (top right), $q>\frac{1}{9}$ (bottom).}
\label{f_potential_Byz}
\end{figure} 
By examining the above transition probabilities,
one can show that~$\hV$ looks as shown 
in Figure~\ref{f_potential_Byz}.
Indeed, let us try to find out
for which values of~$m\leq \tfrac{(1-q)n}{2}$ 
the drift $\tq_m-\tp_m$ would become~$0$ ---
these would correspond to the ``flat points''
(typically minimums or maximums) of the potential~$\hV$.
To this end, we abbreviate $\alpha=\tfrac{m}{n}$, 
and solve the equation
\[
 \alpha \big((1-\alpha-q)^3
 +3(\alpha+q)(1-\alpha-q)^2\big)
 = (1-\alpha-q) \big((\alpha+q\big)^3
 +3(\alpha+q)^2(1-\alpha-q)\big).
\]
Note that one can divide  both sides
by $(1-\alpha-q)$ and then cancel the $\alpha^3$ terms. The two real roots that we are interested in are 
\begin{align*}
    \alpha_0(q) & =
    \tfrac{1}{4}\Big(1-\sqrt{1-\tfrac{8q}{1-q}}\Big)-q,\\
  \alpha_1(q) & =
    \tfrac{1}{4}\Big(1+\sqrt{1-\tfrac{8q}{1-q}}\Big)-q,
\end{align*}
which only exist for $q\in (0,\frac{1}{9}]$.
One can also obtain that 
\begin{equation}
\label{asymp_alpha01}
\alpha_0(q)=3q^2+O(q^3) \text{ and }
\alpha_1(q)=\tfrac{1}{2}-2q+O(q^2) \quad\text{ as } q\to 0.
\end{equation}
Define also the ``symmetric'' points
$\alpha^*_0(q)=1-q-\alpha_0(q)$ 
and $\alpha^*_1(q)=1-q-\alpha_1(q)$.
To be able to distinguish between the situations in
the second and third
pictures in Figure~\ref{f_potential_Byz},
it is then important to be able to compare
the values of $\hV\big(\alpha_0(q)\big)$
and $\hV\big(\tfrac{(1-q)n}{2}\big)$;
for this, we essentially need to know
the sign of the sum $\sum_{m=n\alpha_0(q)}^{(1-q)n/2}
\log \tfrac{\hp_m}{\hq_m}$.
We can approximate the sum with 
the integral and note that
the equation (in~$q$)
\begin{equation}
\label{scary_integral_eq}
 \int\limits_{\alpha_0(q)}^{\frac{1-q}{2}}
  \log\frac{s((1-s-q)^2+3(s+q)(1-s-q))}{(s+q)^3
   + 3(s+q)^2(1-s-q)} \, ds = 0
\end{equation}
has the solution $q^*\approx 0.09029$
(obtained numerically).
Therefore, roughly at~$q^*$ there should
occur the phase transition between
the top left and  top right pictures
in Figure~\ref{f_potential_Byz}.
   
The key difference with the non-Byzantine
case (recall Figure~\ref{f_potential_V})
is that now in this $\tV$-picture there are possibly
three potential wells.
Claim~\ref{cl_behavior_potentials}~(iv)
tells us that, in a situation like that, 
the random walker will typically need quite a lot
of time to escape from the well.

For the 3-choice majority dynamics with the ``help the weakest'' adversarial strategy, we therefore find that
\begin{itemize}
    \item with $q<q^*$, there are three locally attractive states - two pre-consensus ones, and one ``balanced''
(i.e., ``$50/50$''); 
but the pre-con\-sen\-sus states are ground 
(that is, ``stronger''; i.e., in the long run,\footnote{It might be
\emph{really} long, though.}
the system will ``prefer'' to stay in those
due to Claim~\ref{cl_behavior_potentials}~(iii));
\item  with $q^*<q<\frac{1}{9}$, 
those three states remain, but now the balanced one
is the ground state (so the system might spend some time in a pre-consensus state, 
but in the long run it will fall to the ground state,
again due to Claim~\ref{cl_behavior_potentials}~(iii));
\item  with $q>\frac{1}{9}$, there are no pre-consensus states (so the system just goes to the ``balanced''
ground state from anywhere).
\end{itemize}
    
  Let us discuss  the $q<q^*$ situation
 (the top left picture in Figure~\ref{f_potential_Byz}) in a little more detail.
By Claim~\ref{cl_behavior_potentials}~(ii)
(in fact, Lemma~\ref{l_prob_exit} applied to the
present situation), if the process starts sufficiently
out of the interval $[\alpha_1(q)n,\alpha^*_1(q)n]$
(i.e., out of the central potential well), it will go 
to the corresponding pre-consensus state with 
probability exponentially close to~$1$. Then, an important
\emph{practical} observation is that, 
 if in a pre-consensus state, the honest
nodes have a way to figure out the preferred
state of the system, e.g.,\ by averaging over last~$N$ 
responses received,  with a very high probability,
a significant majority will choose the same as the preferred
state. This can be seen as a ``positive'' result:
even with Byzantine nodes, in the case where the initial opinion
configuration is sufficiently away from the ``balanced''
situation, we are able to achieve practical consensus 
with high probability.\footnote{Also, using the technique
of {stochastic domination}, it is even possible to show that 
the same holds for \emph{any} adversarial strategy,
not only for ``help the weakest''.}
However, later we will see
that the presence of the central potential well
does pose some very serious problems in practice.

Let us now ask the following question: what if we slightly increase the communication complexity by asking~$k>3$ nodes
for opinions, instead of just three? 
What are the benefits of this increase?
For convenience, let us assume that~$k$ is odd,
so that draws are not possible. 
First, let us show that 
the distance from total consensus
to the bottoms
of the two
pre-consensus wells will be of order $q^{\tfrac{k+1}{2}}$
as $q\to 0$. 
Recall~\eqref{asymp_alpha01}:
this agrees with the fact that $\alpha_0(q)\asymp q^2$
with $k=3$
. 
Indeed, assume that the system (i.e., the number of honest nodes with current opinion~$1$) is currently in the state
around~$m=Cq^{\tfrac{k+1}{2}}n$. Then, the probability that an honest node with opinion~$1$ is selected 
is approximately $Cq^{\tfrac{k+1}{2}}$, and then it 
will flip its opinion with probability close to~$1$
(since the proportion of the nodes who would respond~$0$
when chosen will be around
$1-q-Cq^{\tfrac{k+1}{2}}>\tfrac{1}{2}$, so 
the fact
that the Binomial distribution becomes more
concentrated as~$k$ increases will take care of that).
On the other hand, the probability  
that an honest node with opinion~$0$ is selected 
is around $1-q-Cq^{\tfrac{k+1}{2}}$, and then it 
will flip its opinion with probability 
around~$c^*q^{\tfrac{k+1}{2}}$
(this can be seen from the explicit formula for the Binomial
probability distribution).
Then, we see that (up to terms of smaller order, and
keeping the notations $\tp_m, \tq_m$ for the case of~$k$
chosen nodes)
$\frac{\tp_m}{\tq_m}\approx \tfrac{C}{c^*}$,
so the bottom of the left pre-consensus well will be 
around $c^*q^{\tfrac{k+1}{2}}n$.
That is, increasing~$k$ indeed permits us to 
approximate the pre-consensus to the consensus,
thus also making it easier to design the local decision rules
used by the honest nodes to finalize their opinions.

However, let us show that the width or size of the central 
potential well will not tend to zero
(in fact, it will always remain at least~$qn$).
Indeed, notice first that if the state of the 
system is close to $\tfrac{(1-2q)}{2}n$, then
(since the Byzantine nodes would vote for~$1$)
a chosen peer would vote for~$0$ or~$1$ roughly
with equal probabilities; however, since 
the probability of selecting an honest 
node with opinion~$1$ is significantly less than
the probability of selecting an honest 
node with opinion~$0$, the process will 
have a drift to the right (meaning that it is
inside the central potential well).
For completeness, observe also that
for a fixed $a<\tfrac{(1-2q)}{2}$ one 
can find a large enough~$k$ such that~$an$
is out of the central well. Indeed,
the probability that a randomly chosen node
would vote for~$1$ would then be 
$a+q<\tfrac{(1-2q)}{2}+q=\tfrac{1}{2}$,
and then it is clear that the probability
that a Binomial$(k,a+q)$ random variable does not
exceed $\tfrac{k-1}{2}$ can be made as small as we 
want by the choice of~$k$. This shows that,
as we grow~$k$, the central potential well
``shrinks towards'' $\big[\tfrac{(1-2q)}{2}n,
 \tfrac{1}{2}n\big]$, but that interval (of length~$qn$) is always
 a part of it. 

To finalize this subsection on a not-so-optimistic
note, observe that we have analyzed only one 
adversarial strategy: the ``help-the-weakest''. 
 This point also
explains why analyzing the system via simulations is not a straightforward task --- these need to be 
performed for \emph{any} adversarial strategy,
and there are infinitely many of them.
It is also necessary to mention that, 
in principle, with more complicated adversarial strategies
or node's finalizations rules 
the process may not even be a Markov chain
which would make its analysis much more difficult.

\subsection{Obstacles on the road to consensus:
the curse of metastability}
\label{s_curse_met}
We have just seen the most immediate consequence of having
Byzantine nodes: some natural adversarial strategies can lead 
to metastable states, where, due to
Claim~\ref{cl_behavior_potentials}~(iv),
the process can spend quite a lot of time
(exponential in the number of Byzantine nodes
provided that they make a positive fraction
of the totality, which could mean
essentially forever in practice).
\begin{figure}
\begin{center}
  \includegraphics[width=\textwidth]{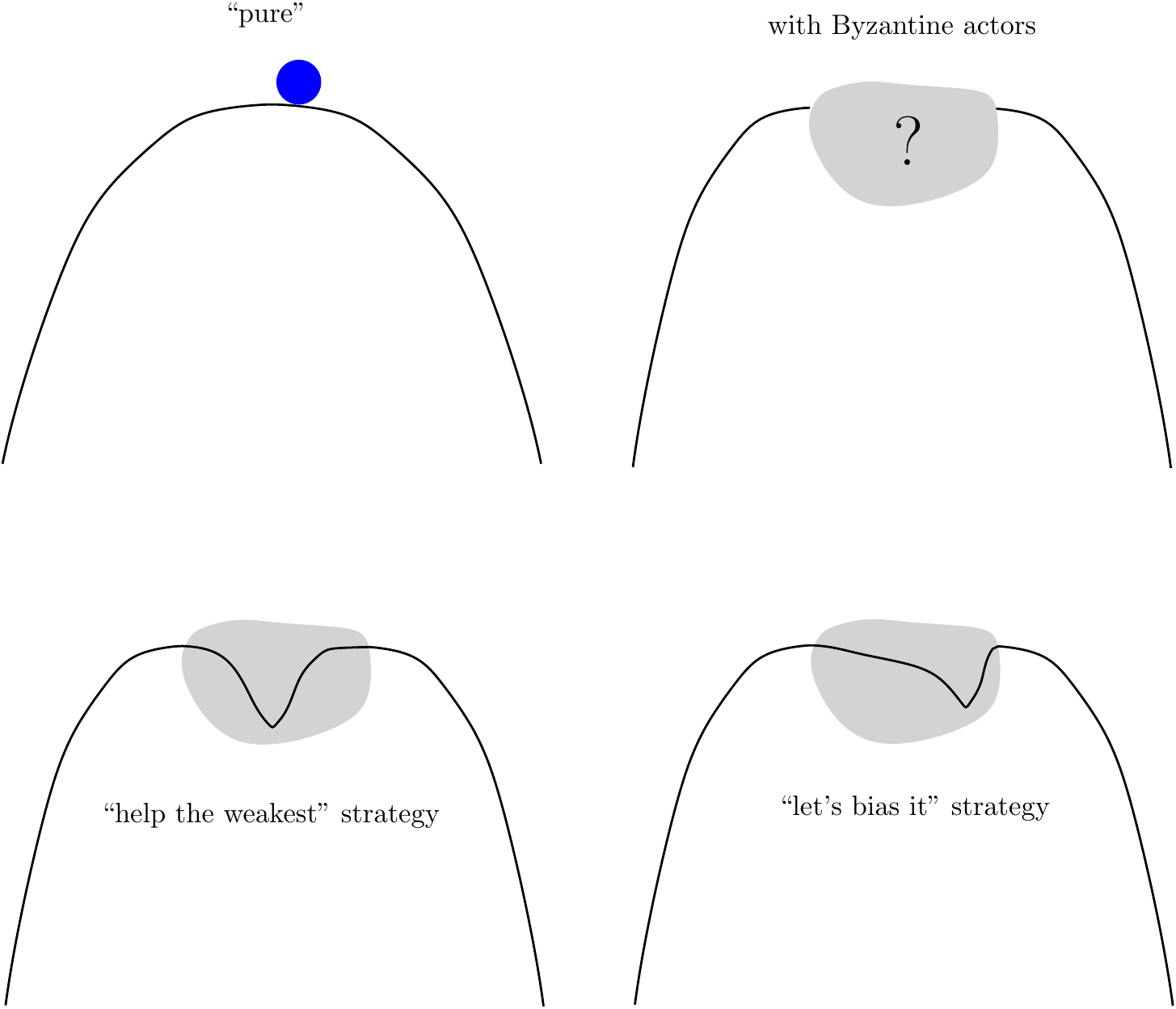}
\end{center}
 \caption{Majority dynamics with Byzantine actors:
the curse of metastability.}
\label{f_maj_dyn}
\end{figure} 
%
The situation is summarized in
Figure~\ref{f_maj_dyn} (we disregard the pre-consensus 
potential wells there, as they may be made small by the choice
of~$k$, and, in any case, the honest nodes are able to make the correct
decision while the process is spending time in those).
Essentially, in the all-honest situation (top left picture),
regardless of the starting position, the process will
quickly end up in one of the two consensus states.
However, with Byzantine actors (top right picture),
the adversary has a ``region of control''
(indicated by the big question mark) where it
is able to dictate what the system will do.
For example, for the natural ``help the weakest''
strategy (discussed in detail in Section~\ref{s_enter_Byz})
the potential profile is as in the bottom left picture,
meaning that, when started not far away from the balance,
the process is likely to spend a lot of time near the 
``exactly balanced'' situation. As noted at the end of the previous
section, other adversarial strategies are possible,
for example, the one on the bottom right part of Figure~\ref{f_maj_dyn}.
There, the adversary prefers to ``bias'' the overall
opinion counts slightly towards one of the extremes
(while still remaining in its region of control).
This approach might make sense because it would increase 
the probability that, during that long time spent in 
the potential well, 
some of honest nodes would already 
finalize on that ``slightly preferred'' opinion
because their finalizations rules would be triggered 
by some ``improbable'' events
(and then the adversary would even be able to sway the consensus
to the other side).

The last point highlights an important (and perhaps
not a-priori clear) problem: the existence of metastable
states not only compromises liveness but also casts
doubts on the system's safety.
The reason is that, as it sometimes happens with 
large stochastic systems, improbable things can happen
if one waits for a really long time. 
Since the local finalization rules used by the nodes
are usually publicly known, the adversary may try
using that knowledge to create conditions whereby
 finalization still might occur during a sufficiently
long time interval while the process is still
within the adversary's region of control. 
To assess the risk of such a development,
one would need to employ some entropy-versus-energy-type arguments (that,
basically, compare the inverse probability of an ``improbable''
event to the number of ``tries'' when that event might occur),
and the question of optimizing over adversary's strategies
would stand even then. 

In view of the above, a natural idea would be to set
some limit on the time to achieve consensus, i.e.,
if a transaction's status is still undecided after some
fixed time~$t_{\text{max}}$, then just consider this transaction invalid, i.e., set its status to~$0$.
However, this can be dangerous, because some nodes would 
finalize on~$1$ just barely while others just barely not
(and therefore set the final opinion to~$0$ according 
to that rule), with
a non-negligible probability (especially if the adversary
stops messing with the system ``just before''
the limit time~$t_{\text{max}}$ elapses).
Another natural idea would be to introduce some 
tie-breaking rule: based on the information 
available in the system,
 the nodes would spontaneously
switch to one of the opinions 
in case the balanced state persists for too long.  If enough honest nodes 
do that approximately at the same time,  the system may be pushed
towards one of the extremes.
This idea certainly deserves further research;
however, this sort of metastability-breaking is not easy
to achieve when the adversary has (almost) full information: 
the adversary would be facing a sort of stochastic control problem, 
and might be able to influence the state of the system
(i.e., ``adjust the controls'') in such a way
that not enough honest nodes are able to take that 
sort of coordinated action, and therefore the split would persist.
In the next section, however, we will describe a 
metastability-breaking mechanism that uses \emph{external
randomness} (which cannot be \emph{predicted}
by the adversary) and show that it indeed makes
the system achieve consensus in a well-controlled way.

In any case, to explore the attacks outlined in this
section, the adversary has to be ``smart''
(i.e., adjust its behavior to the state of the honest nodes
in perhaps a non-trivial way)
and well-connected (to be able to know the state
at least with some degree of precision).
It is clear that the system will not ``break by itself''
(with the adversary not abiding by the rules but just doing random things); on the other hand, it is also clear
that a smart-and-well-connected adversary can disrupt
the system to a considerable extent, at least in theory.
For sure, one has to bear in mind that the adversary's omniscience
assumption is an idealized one, and that in a practical system
 from the adversary's point of view there will be some 
``natural randomness'' (stemming from, e.g., differences in nodes' local perceptions, 
fluctuations in message delivery times, own limitations
in connectivity and processing power that would prevent the adversary from
catching up with the rest of the network, 
and so on) that the adversary
cannot predict/control. 
In principle, this ``adversarial unawareness''
could also act as a metastability breaker, but one would have to exercise 
quite some caution when \emph{relying} on it. The problem is 
that there will likely be a sharp phase transition with respect
to the amount (whatever it means) of that randomness:
too much ``fog of war'' would prevent the attacker from maintaining metastable
states, but a small amount of it would mean that the above-discussed  potential
wells are still present and therefore the adversary would be still in control (in the situation when opinions are split).
In any case, as we have already mentioned, it is not necessarily clear how to 
define/measure/control this factor and what  guarantees there are that 
no adversary would be able to become \emph{sufficiently} powerful
in the foreseeable future.


One example of a paper that could have benefited from
the analytic tools and arguments outlined here
is~\cite{Ava19}
(we cite also~\cite{Ava18}, which contains more
technical details).
As argued above,
to access the protocol's safety it is not
enough to just prove that it will converge
to consensus starting from an already sufficiently
biased situation
(i.e., sufficiently out of the central potential
well, in the case of the ``help the weakest''
adversarial strategy; admittedly, even for this
situation the arguments in these papers are 
far from convincing).
There is no attempt to approach it via the 
above-mentioned ``entropy-versus-energy'' 
arguments, which, in particular, would be important
in the case of the above-mentioned
``let's bias it'' strategy. 
The authors of the aforementioned papers also claim that they have studied
their model via simulations, but it is not completely clear 
which adversarial strategies were considered,
and why those strategies would be optimal from the adversary's
point of view.
It is also unfortunate that these papers
do not contain explicit estimates on the probabilities
of ``bad'' events;
as argued in this section,
in such circumstances it is hardly
possible to access if the system can scale well.
As consensus protocols play a prominent role in every distributed ledger's safety,
the authors want to stress the need to approach any ``solution'' without a metastability-breaking mechanism with skepticism.

\section{FPC: breaking metastable states}
\label{s_fpc}
In this section, we describe a recent solution, the fast probabilistic consensus (FPC) introduced
in~\cite{fpc}. FPC breaks metastable states with the use of \emph{external randomness}. 

The use of common external randomness necessitates a change of the dynamics. We no longer ask nodes to update their opinions one by one, but instead ask them to update their opinions synchronously in rounds. 
As previously noted, it is possible that
such a synchronous model could \emph{in theory} be analyzed as the asynchronous toy model, but the analysis would be much more technical. 
In fact, the corresponding random walk would not only be able to jump steps of length $1$, but its step distribution may be of full support.

As in the previous toy model, the basic feature of FPC is that, in each round, each node queries a random subset of known peers of size $k$ about their current opinion.  We allow~$k$ to be relatively large, e.g., $k=25$, but still assume that $k\ll n$.
Once a node receives the answers to its queries, it updates its opinion. This continues until a local stopping rule is satisfied, as discussed in the previous section. 

On a high level, the main idea of the FPC can be explained in 
the following way. As we have seen in Section~\ref{s_maj_Byz},
if the adversary is able to know (at least with some
precision) what the state of the honest part of the system is
and is able to predict the honest nodes' behavior, 
then it may be able to seriously interfere with the system.
\begin{figure}
\begin{center}
 \includegraphics{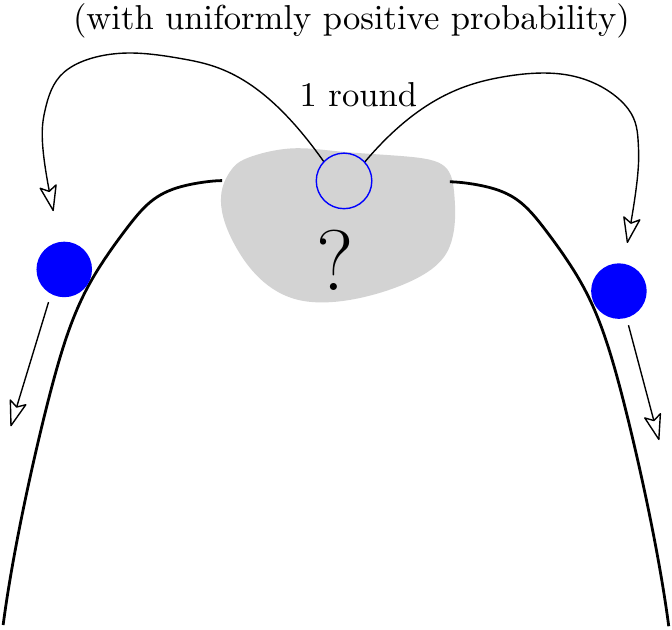}
\end{center}
 \caption{FPC: defeating the metastability 
 with round-based common random thresholds.}
 \label{f_FPC_idea}
\end{figure} 
On the other hand, assume that the decision rule used by 
the honest nodes to choose the preferred opinion is 
\emph{unpredictable}; that is, cannot be known in advance
by the adversary.\footnote{Or, at least, \emph{with positive
probability}
the adversary is not able to predict it.} 
In this case, the result is illustrated in Figure~\ref{f_FPC_idea}:
there is a good chance that, in just one round,
the situation completely escapes the adversary's control
(the system jumps into a pre-consensus state from where
it will quickly converge to consensus regardless
of the adversary's actions). This unpredictability
is achieved by using a sequence of common random numbers,
which are either provided by an external source or generated
by the nodes of the system themselves
(or, in practice, by a subcommittee of those),
see e.g.,~\cite{cascudo2017scrape, Lenstra_Wes17, popov2017decentralized, 
schindlerhydrand, syta2017scalable, wes}.

\subsection{Notation and description of the protocol}
\label{s_notation_fpc}
Every node $i$ starts with an initial state or opinion $s_{i}(0)\in \{0, 1\}$. We update every nodes opinion at each round, and we note $s_{i}(t)$ for the opinion of the node $i$ at time $t\in \mathbb{N}$. 
At each time step each node chooses $k$ (random) neighbors $C_{i}$, queries their opinions, and calculates 
\begin{equation*}
\eta_{i}(t+1)=\frac1{k_{i}(t)} \sum_{j\in C_{i}} s_{j}(t),
\end{equation*}
where $k_{i}(t)$ is the number of replies at time $t$.

It is possible that not all nodes respond and one can set the value for $s_j(t)$ in the case of no response to $0$, $1$ or \verb?NA?. This case is similar to a proportion of nodes being faulty and can be modeled through strategies for malicious nodes introduced in Section~\ref{sec:maliciousNodes}.

In FPC the first round may be different from the subsequent rounds.  Any binary decision, as in statistical test theory or binary classification, gives rise to two possible errors. In many applications, the two possible errors are asymmetric in the sense that one error is more severe than the other.  
Choosing the threshold of the first voting round differently allows the protocol to include this kind of asymmetry. Let $U_1$ be a uniform distributed random variable $\mathrm{Unif}( [a, b])$, for $0 \leq a\leq b \leq 1 $,  and $U_{t}, t= 1, 2,\ldots$ be uniform random variables with law $\mathrm{Unif}( [\beta, 1-\beta])$ for some parameter $\beta \in [0,1/2]$. We assume the random variables $(U_t)_{t\geq 1}$ to be independent.

The update rules are now given inductively  for $t\geq 1$ by
\begin{equation}\label{eq:Update}
s_{i}(t+1)=\left\{ \begin{array}{ll}
1, \mbox{ if } \eta_{i}(t+1) > U_{t}, \\
0, \mbox{ if } \eta_{i}(t+1) < U_{t}, \\
s_{i}(t), \mbox{ otherwise.}
\end{array}\right. 
\end{equation}

The essential ingredient of FPC, compared to previous voting models is that the thresholds $U_t$ are random. The usual majority dynamics would compare the $\eta$'s against the deterministic threshold~$\tfrac{1}{2}$.

Nodes have to decide locally when to stop the protocol and  no longer update their opinions. One possibility of such a local stopping rule is that a node stops querying other nodes and finalizes its opinion 
if it did not change its opinion for the last~$\ell$ rounds.
In the original paper~\cite{fpc} there is also a \emph{cooling-off}
period of~$m_0$ rounds after which the above rule starts to be applied. 
A node will continue to answer queries until some maximal number of rounds, say~$\Delta$, is attained.

We will distinguish in this exposition between  two  different ways a voting-based consensus protocol can fail.

\begin{enumerate}
\item \textit{Termination failure}: 
We say the protocol suffers a termination failure if some node did not stop querying before the maximal number of rounds. 

\item \textit{Agreement failure}: 
We say the protocol suffers an agreement failure if two honest nodes terminate with different final opinions. 
\end{enumerate}

We say that the protocol \emph{terminates in agreement}, 
if there are neither agreement nor termination failures.

\subsection{Threat model}
\label{sec:maliciousNodes}
Byzantine infrastructures and attack strategies can be very diverse. We assume that adversarial nodes can exchange information freely between themselves and can agree on a common strategy. In fact, they all may be controlled by a single individual or entity. 

As in Section~\ref{s_enter_Byz},
we suppose that $q n$ nodes are adversarial for some $q\in [0,1]$. The remaining $(1-q)n$ nodes are \emph{honest}, i.e., they follow the recommended protocol.

With respect to the \emph{behavior} of the adversarial nodes, we distinguish three cases:
\begin{itemize}
 \item \emph{Cautious adversary}\footnote{Also known as a \emph{covert adversary}, cf.~\cite{aumann2007security}.}: 
Any adversarial node responds the same value to all the queries it receives in that round. The adversary knows the opinions of the honest nodes in the previous rounds.
\item \emph{Semi-cautious adversary}: 
Any adversarial node does not give \emph{contradicting} responses (i.e., $0$ to one node
and~$1$ to another node in the same round) but does not have to answer all queries. The adversary knows the opinions of the honest nodes in the previous rounds.
 \item \emph{Berserk adversary}:
 Any adversarial node may respond differently to different queries in the same round. At each moment of time, the adversary is not only aware of the current opinions of all honest nodes but may also know which nodes query which other nodes.
\end{itemize}
The berserk adversary is the worst-case scenario; we pose no limitation on the node's behavior and allow it to be omniscient, i.e.,  it knows all information that exists until \emph{now}. However, it is not \emph{prescient} nor possesses prior knowledge or influence on the random threshold. 

Let us explain why the adversary may choose to be cautious or semi-cautious. We assume that nodes have identities and sign all their messages; this way, one can always \emph{prove} that a given message originates from a given node.
Now, if a node is not cautious, 
this may be detected by the honest nodes, e.g., two  honest nodes may exchange their query history and verify that the same node passed contradicting information to them.
In such a case, the offender may be penalized
by the honest nodes. For example, the nodes who discovered
the fraud would pass that information along with the relevant proof and the other honest nodes would stop querying this particular node.

The semi-cautious adversary has slightly more freedom in \emph{lying} (in fact, in not saying anything) than a cautious adversary. We will see later that this additional freedom has an impact on the security of the system.


\subsection{Theoretical results}
\label{s_theory_fpc}
The theoretical results on FPC found in~\cite{fpc} are obtained under various assumptions. We discuss
the necessity of these assumptions in Section~\ref{sec:discussion}. 
In this section we assume that every node knows all other participants, that it can  \emph{directly} query any other node, and that honest nodes always reply in due time. In addition, we assume the existence of a common source of randomness that delivers a fresh random threshold in every round.

A complete formulation and proof of the theoretical results of~\cite{fpc} would require some more notations and end up being rather technical. 
For this reason, we present several results from~\cite{fpc} that stand by way of example for the more general and more technical results. 

Let us define two events relative to the final consensus
value:
\begin{equation}
 H_i = \{ \text{all honest nodes eventually reach
   final opinion } i\},
\end{equation}
$i=0,1$.
Thus, the union $H_0\cup H_1$ stands for the event
that all honest nodes agree on the same value,
i.e., that the consensus was achieved.

Let $\mathcal{N}$ be the number of rounds until 
\emph{all} honest nodes achieve their final opinions.
In the following result
(which is a corollary of Theorems~4.1 and~6.1
of~\cite{fpc}), we obtain a lower bound 
on the probability of the event
$(H_0\cup H_1)\cap\{\mathcal{N}\leq m_0+\ell \}$.
This event means that the consensus was achieved 
in the shortest possible time~$m_0+\ell$. 

\begin{theo}
\label{t_main}
\begin{itemize}
 \item[(i)] In the case of a cautious adversary, assume that
 $q<1/2$. Then, for any fixed $\beta\in (q,1/2)$
 we have, with $c_{0,1,2}>0$ depending on~$\beta$ and~$q$
 \begin{align}
\lefteqn{
\IP\big[(H_0\cup H_1)\cap\{\mathcal{N}\leq m_0+\ell \}\big]
 } \nonumber\\
& \geq 1-c_0n\ell \exp(-c_1 k) - \exp(-c_2 m_0\ln k) .
 \label{eq_safety_liveness_cautious} 
\end{align}
\item[(ii)] In the case of a semi-cautious adversary, assume that
 $q<\frac{3-\sqrt{5}}{2}$.
 Then,\footnote{Note that  $\frac{3-\sqrt{5}}{2}=\frac{1}{1+\phi}=\phi^{-2}\approx 0.38$,
 where $\phi=\frac{1+\sqrt{5}}{2}$ is the Golden Ratio.} 
 for any fixed $\beta\in (q,\frac{1-q}{2-q})$
 we have, with $c_{0,1,2}>0$ depending on~$\beta$ and~$q$
 \begin{align}
\lefteqn{
\IP\big[(H_0\cup H_1)\cap\{\mathcal{N}\leq m_0+\ell \}\big]
 } \nonumber\\
& \geq 1-c_0 n\ell \exp(-c_1 k) - \exp(-c_2 m_0) .
 \label{eq_safety_liveness_semi} 
\end{align}
 \item[(iii)] In the case of a berserk adversary, assume that
 $q<1/3$. Then, for any fixed $\beta\in (q,\frac{1-q}{2})$
 we have, with $c_{0,1,2}>0$ depending on~$\beta$ and~$q$
 \begin{align}
\lefteqn{
\IP\big[(H_0\cup H_1)\cap\{\mathcal{N}\leq m_0+\ell \}\big]
 } \nonumber\\
& \geq 1-c_0n\ell \exp(-c_1 k) - \exp(-c_2 m_0) .
 \label{eq_safety_liveness_berserk} 
\end{align}
\end{itemize}
\end{theo}

It is also important to observe that the above theoretical
estimates (as well as those of~\cite{fpc})
are not necessarily sharp; this explains why we keep 
the possibility of using a generic~$\ell$ even
though $\ell=1$ would be the most ``favorable''
for the above 
estimates~\eqref{eq_safety_liveness_cautious}--\eqref{eq_safety_liveness_semi};
in practice a higher value of~$\ell$ may work better (as shown in~\cite{fpcsim}).

Asymptotic results are of interest in evaluating the scalability of the protocol. 
In addition, the next corollary gives more quantitative details on the probability of terminating in an agreement.
\begin{cor}\label{cor:asympt}
\label{c_log_n}
Let $\beta=1/3$ and $\ell$ be some constant. We assume that the proportion of Byzantine nodes~$q$
is \emph{acceptable}, i.e., less than~$1/2$ for 
the case of a cautious adversary, less than $\frac{3-\sqrt{5}}{2}$ for a semi-cautious adversary,  or less than~$1/3$
for the case of a berserk adversary. Then, there exists some constant $C$ such that for $k=C \log n$ and
\begin{itemize}
    \item $m_0 = O\big(\frac{\ln n}{\ln \ln n}\big)$ for the cautious adversary,  
    \item $m_0=O(\ln n)$ for the semi-cautious and berserk adversary,
\end{itemize}
the probability of terminating in agreement is at least $1-\eps(n)$, where~$\eps(n)$
is polynomially small in~$n$ (i.e., $\eps(n) = O(n^{-h})$ for 
some~$h>0$, and this~$h$ can be made arbitrarily large by choosing a suitable~$C$). 

In particular, the overall communication complexity is at most
$O\big(\frac{n \ln^2 n}{\ln \ln n}\big)$
for a cautious adversary and $O(n \ln^2 n)$ for a berserk one.
\end{cor}

Now, we outline the strategy of the proof of the above results.
Let~$\hp_m$ be the proportion of $1$-opinions
among the honest nodes after the $m$th round.\footnote{This notation
does not have anything to do with $\hp_m$s of Section~\ref{s_enter_Byz}.}
Let us define the random variable
\begin{equation}
\label{def_Psi}
 \Psi = \min\big\{m\geq 1: \hp_m\leq \tfrac{\beta-q}{2(1-q)}
 \text{ or }\hp_m\geq 1-\tfrac{\beta-q}{2(1-q)}\big\}
\end{equation}
to be the round after which 
the proportion of $1$-opinions among the honest
nodes either becomes ``too small'', or ``too large''.\footnote{One
may think of~$\Psi$ as the first time the systems enters
the ``very steep'' part of the potential, 
as explained in the previous section.}
The idea is that, when~$\Psi$ happens, it then becomes extremely likely 
that all~$(1-q)n$ honest nodes will maintain the same 
opinion (which will be $0$ if the first condition
in~\eqref{def_Psi} occurs or $1$ if the second one does)
during the~$\ell$ consecutive rounds needed for reaching consensus. 
It is here that the fact that $q<1/2$ is important: when 
the honest nodes are almost in agreement, the adversarial nodes
do not have the necessary voting power to convince them otherwise.
In fact, it is also important that the parameter~$\beta$
is chosen in such a way that $\beta>q$:
in the situation when almost all honest nodes agree,
the adversary must not 
be able to reach the interval $[\beta,1-\beta]$ where
the decision thresholds live.
As the reader probably figured out, the first negative
term in the right-hand sides 
of~\eqref{eq_safety_liveness_cautious}--\eqref{eq_safety_liveness_semi}
correspond to the probability that at least 
one node of the system ``jumps out''
of the consensus state during the~$\ell$ finalization rounds (in particular, the factor $n\ell$ there comes from the union bound).
We do not elaborate further in this paper; a rigorous derivation
of the corresponding bounds was done in~\cite{fpc}.

Now let us explain the origin of the second 
negative terms in the right-hand sides 
of~\eqref{eq_safety_liveness_cautious}--\eqref{eq_safety_liveness_semi};
those come from the tail estimates on~$\Psi$
-- the probabilities that~$\Psi$ did not happen during the~$m_0$ preliminary rounds. In the following, we elaborate on this
for cautious, berserk, and semi-cautious adversaries
(and also explain the reasoning behind the corresponding
security thresholds $1/2$, $1/3$, and $\phi^{-2}=\frac{3-\sqrt{5}}{2}$).
To simplify the arguments, we assume that the honest nodes' opinions
are roughly equally split (that is, approximately $\frac{1}{2}(1-q)n$
of those have current opinion~$0$ and approximately $\frac{1}{2}(1-q)n$
 have current opinion~$1$), and the adversary is attempting to maintain
this 50/50 split.

\paragraph{Cautious adversary} This is the simplest case since the adversary essentially has to choose the opinions 
of the nodes it controls beforehand; this in turn means
that the \emph{expected} proportions of $1$-opinions that honest 
nodes receive will be the same, and therefore the 
\emph{actual} proportions of $1$-opinions they receive
will likely be distributed on an interval of length~$O(k^{-1/2})$
because of the Central Limit Theorem
(we refer the reader to Figure~2 of~\cite{fpc}).
To maintain the opinions split, the next (not yet known to the adversary)
decision threshold has to be in that interval; this happens
with a probability of order~$k^{-1/2}$.
Therefore, this explains the second negative term in the right-hand 
side of~\eqref{eq_safety_liveness_cautious}
(note that $(ck^{-1/2})^{m_0}=\exp(-\frac{1}{2}m_0(\ln k -2\ln c))$).
Note also that here we do not need further restrictions
on~$q$ (other than $q<1/2$). 

\paragraph{Berserk adversary} In this case, the adversary can 
adopt for example the following strategy: feed $0$-opinions to one-half
of the honest nodes, and $1$-opinions to the other half.
Then the honest nodes from the first half will receive 
around $\frac{1}{2}(1-q)$ proportion of $1$s (because they 
will only receive $1$-opinions from honest nodes),
while the proportion of $1$s that
those from the second half will receive 
will be around $\frac{1}{2}(1-q)+q=\frac{1}{2}(1+q)$.
Then, to assure that~$\Psi$ happens with at least
a positive constant probability\footnote{This gives rise to 
the second negative term in the right-hand side of~\eqref{eq_safety_liveness_berserk};
as before, it corresponds to the probability of~$\Psi$ not happening during~$m_0$
consecutive rounds.},
we need $[\frac{1}{2}(1-q),\frac{1}{2}(1+q)]$
to be a proper subset of $[\beta, 1-\beta]$
(since in this case with uniformly positive probability
the decision threshold will be outside of 
$[\frac{1}{2}(1-q),\frac{1}{2}(1+q)]$, and therefore most 
of the honest nodes will adopt the same opinion in the next 
round).
This means that, besides the initial restriction~$q<\beta$,
we also need to have $\frac{1}{2}(1-q)>\beta$, that is,
$\beta\in (q,\frac{1}{2}(1-q))$. The last interval
must be nonempty meaning that $q<\frac{1}{2}(1-q)$,
which is equivalent to $q<1/3$.

\paragraph{Semi-cautious adversary} 
Recall that we are assuming that the current opinions
of honest nodes are roughly equally split; we will also
consider only a \emph{symmetric} adversarial strategy:
$\frac{1}{2}qn$ adversarial nodes will answer~$0$
or remain silent while the remaining $\frac{1}{2}qn$ 
adversarial nodes will answer~$1$
or remain silent
(the general case is treated in Section~6 of~\cite{fpc}).
Let us see what happens if the adversary adopts a similar strategy
as considered above for a berserk one: to  half the honest nodes
the adversarial nodes will answer~$0$ (if possible) or remain 
silent, while to the other half they 
will answer~$1$ (if possible) or remain silent.
Note that a random query of any honest node will be answered
with probability $1-q+\frac{q}{2}=1-\frac{q}{2}$ in this setup.
Then, given that an honest node from the first half receives
an answer, this answer will be
\[
 \begin{cases}
  0, & \text{with probability } 
  \frac{\frac{1}{2}(1-q)+\frac{q}{2}}{1-\frac{q}{2}}=\frac{1}{2-q},\\
  1, & \text{with probability } 
  \frac{\frac{1}{2}(1-q)}{1-\frac{q}{2}}=\frac{1-q}{2-q},
 \end{cases}
\]
and, likewise, given that an honest node from the second half receives
an answer, this answer will be~$0$ with probability~$\frac{1-q}{2-q}$
and~$1$ with probability~$\frac{1}{2-q}$.
So, roughly speaking, the adversary can achieve the following:
the proportion of $1$-opinions received by the honest nodes from the first half
will be around\footnote{As usual, with random fluctuations typically of 
order~$k^{-1/2}$.} 
$\frac{1-q}{2-q}$, while 
the proportion of $1$-opinions received by the honest nodes from the second half
will be around~$\frac{1}{2-q}$. 
Therefore, similarly to the berserk case, we need the 
interval $[\frac{1-q}{2-q},\frac{1}{2-q}]$ to be a proper subset 
of~$[\beta,1-\beta]$ which (again, together with the condition $q<\beta$)
shows that the inequality $q<\frac{1-q}{2-q}$ must hold.
Solving this inequality for $q\in [0,1]$, 
we obtain~$q<\frac{3-\sqrt{5}}{2}$.

Let us also comment on the cooling-off period~$m_0$: 
it was initially introduced in~\cite{fpc} with the following
idea in mind:
let us first give some time to~$\Psi$ to happen,
and then the nodes can reach the consensus in safety. 
In fact, roughly the same asymptotic results as in Theorem~\ref{t_main}
and Corollary~\ref{c_log_n} can be proved also for 
the protocol's version with no cooling-off; specifically,
for the probability 
$\IP\big[(H_0\cup H_1)\cap\big\{\mathcal{N}\leq \frac{3}{2}\ell \big\}\big]$ the same estimates with $m_0$ substituted to~$\ell$
(and possibly different constants) will be valid. To see that, first observe 
that the probability that~$\Psi$ does not happen during the first~$\tfrac{\ell}{2}$ rounds will be at most 
$\exp\big(c\tfrac{1}{2}\ell \ln k\big)$ for the cautious 
adversary and 
at most $\exp\big(c\tfrac{1}{2}\ell \big)$ for
the semi-cautious and berserk ones
(recall the last terms in~\eqref{eq_safety_liveness_cautious}--\eqref{eq_safety_liveness_berserk}).
Then, every honest node would need to decide on the majority
opinion between $\tfrac{\ell}{2}$ and~$\ell$ consecutive times;
the middle terms in the right-hand sides 
of~\eqref{eq_safety_liveness_cautious}--\eqref{eq_safety_liveness_berserk} still are upper bounds to the probability
that the above does not happen.

\subsection{Numerical results}
\label{s_numerical_fpc}
The theoretical results are strong in the sense that they cover all possible adversary strategies and that they give asymptotic results. However, the achieved bounds are not optimal and the actual performances of FPC seem to be much better. In this section, we highlight some results obtained by Monte-Carlo simulations that illustrate the performances of FPC and also indicate that most of the strict assumptions for the theoretical results can be considerably weakened. We also describe a  concrete berserk strategy that is able to break the usual majority-dynamics without the random threshold.

FPC is governed by many parameters and allows users to adjust the protocol in many ways. In~\cite{manaFPC} various different parameter setting are studied with respect to the protocol performances. The authors of \cite{manaFPC} proposed a simplified model and removed the randomness of the initial threshold and the cooling-off phase. While the theoretical results of~\cite{fpc} show that the cooling-off phase may reduce the minimum number of voting rounds required for termination, the asymptotic results in Corollary~\ref{cor:asympt} states that $m_0$ can be chosen as $O(\log n)$. Moreover, the proof strategy in~\cite{fpc} and the discussion at the end of the previous section suggest that the theoretical results remain valid in the absence of the cooling-off phase.  
While in the original FPC, the initial threshold is randomly chosen in an interval $[a,b]$, we want to emphasize that the theoretical results in~\cite{fpc} remain valid for a deterministic threshold and that the randomness in the subsequent rounds is sufficient for the protocol's robustness. For these reasons, we set in this section $a=b=\tau$. 

As before, we assume that a proportion of $q\in[0,1)$ nodes are malicious and try to interfere with the protocol.
As mentioned earlier  we distinguish between  different kinds of adversarial behavior. In the following, we give two explicit examples of possible adversary strategies.

\paragraph{Cautious strategy for agreement and termination failure}
We consider the cautious strategy where the adversary transmits in round $t+1$ the opinion of the minority of the honest nodes of round $t$. This is the ``help-the-weakest'' strategy of Section~\ref{s_enter_Byz}. 
In~\cite{fpcsim} this strategy is dubbed the \textit{inverse vote strategy} (\textit{IVS}).

\paragraph{Berserk strategies for agreement and termination failure}

We consider the berserk strategy known as \textit{maximal variance strategy} (\textit{MVS}). In this approach, the adversary waits until all honest nodes received opinions from all other honest nodes. The adversary then tries to subdivide the honest nodes into two equally sized groups of different opinions while trying to maximize the variance of the $\eta$-values (recall Section~\ref{s_notation_fpc}).

In order to maximize the variance, the adversary requires knowledge over the $\eta$-values of the honest nodes at any given time. The adversary then answers queries of undecided nodes in such a way that the variance of the $\eta$'s is maximized by keeping the median of the $\eta$'s close to $0.5$; 
we refer to~\cite{fpcsim} for a pseudo-code of the attack strategy. Intuitively, this strategy tries to make the central well in the upper left potential in Figure~\ref{f_potential_Byz} as deep as possible.

The effectiveness of the random threshold is clearly seen when comparing the evolution of the different $\eta$ values for each node.  Figure~\ref{fig:eta-heatmap2},  from~\cite{manaFPC}, shows the different evolution of IVS and MVS. In the two upper graphs, the threshold is deterministic, i.e., $\beta=0.5$. We can see in a) that the cautious adversary following IVS with $q=0.3$ can maintain the system in a metastable situation for some time, but eventually converges to the all $0$ situation.  In b) the berserk MVS strategy with $q=0.1$ can keep the system in a meta-stable situation for the whole duration of the simulation. Part c) shows the effectiveness of the random threshold; the berserk attacker looses control very fast. 

\begin{figure}
\begin{center}
  \includegraphics[width=.9\textwidth]{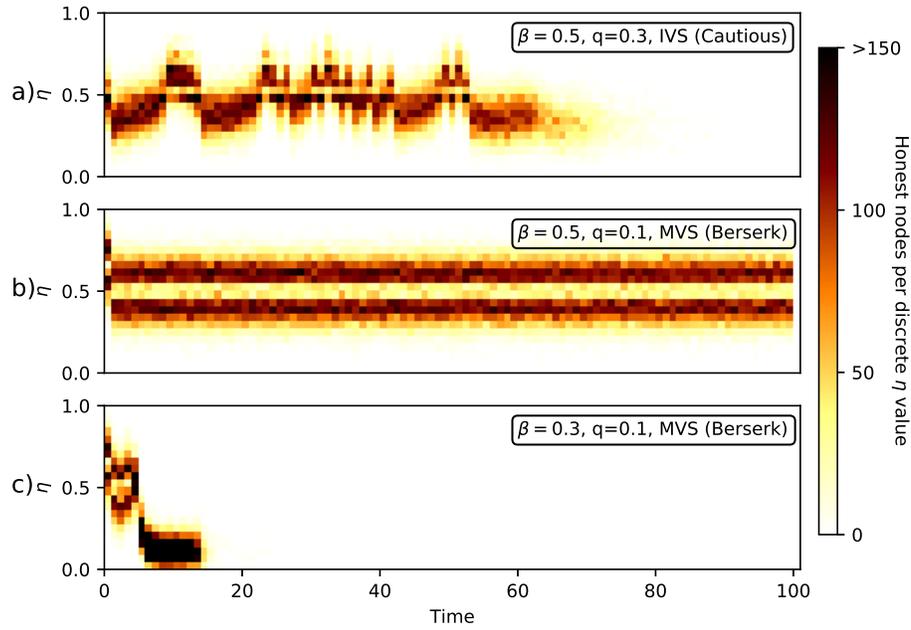}
  \vspace{-0.5cm}
  \caption{Evolution of the number of undecided nodes for a given $\eta$-value; \cite{fpcsim}.}
  \label{fig:eta-heatmap2}
  \end{center}
\end{figure}

Figures \ref{fig:q-beta-T} and \ref{fig:q-beta-A} show the variation of the termination and agreement rate with $q$ and $\beta$. The berserk attacker follows the MVS; and the network size is $n=1000$, $a=b=2/3$ and $\ell=10$. The maximal number of rounds is set to $100$. We consider the worst-case scenario $\hp_0=a=b$.  In cases where $\hp_0$ differs from the initial threshold, the performances of FPC are significantly better. Recall that, if the protocol already starts in a ``steep'' part of the ``potential'' an attacker has practically no influence. It can be seen that if no randomness is employed the protocol can be prevented from terminating. Introducing randomness via decreasing $\beta$ improves the termination as well as the agreement rate. However, the figures also show that too much randomness can harm the performances.  
\begin{figure}
\begin{center}
\begin{minipage}{0.45\textwidth}
    \includegraphics[width=1.3\textwidth]{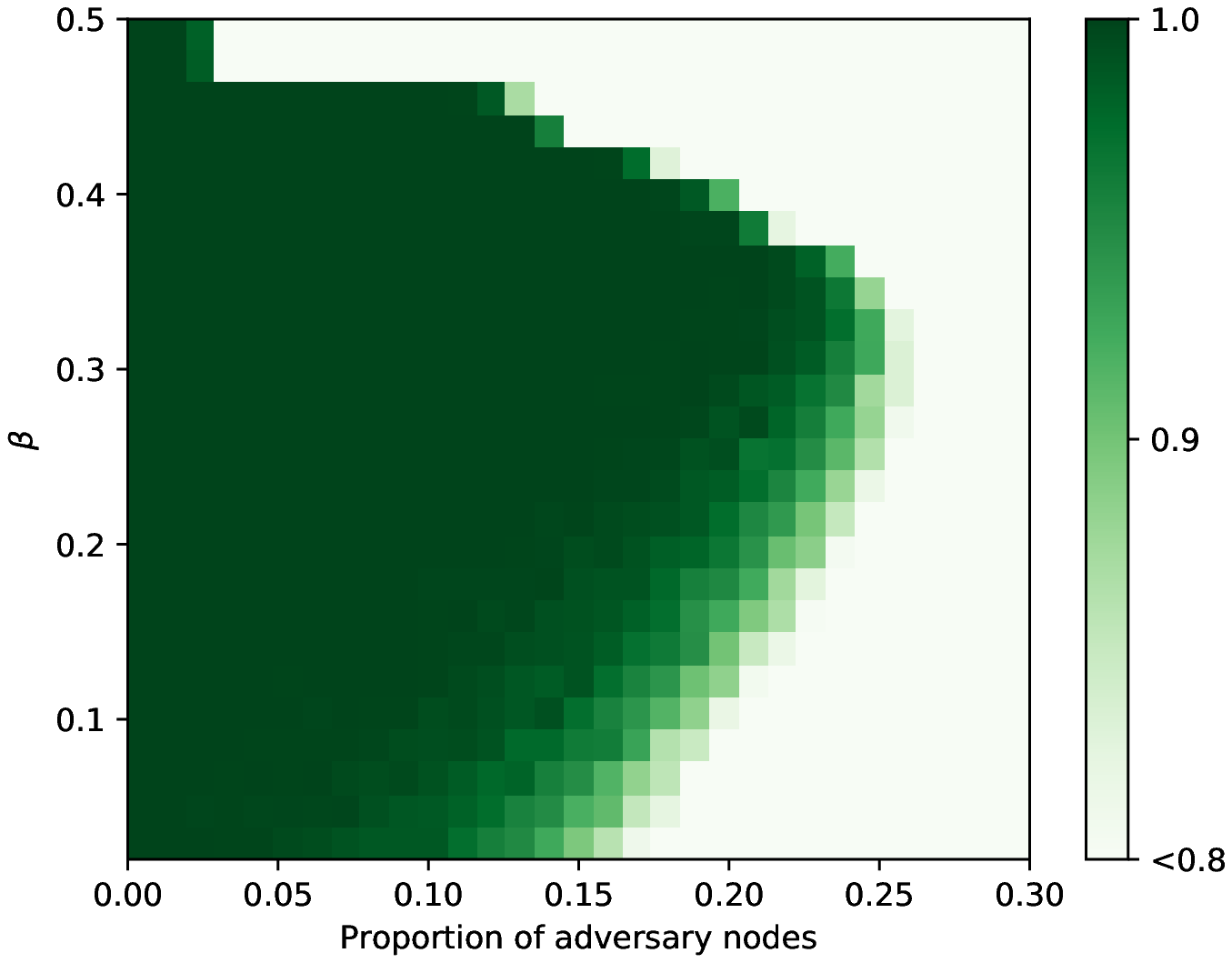}
    \vspace{-0.5cm}
    \caption{Termination rate; \cite{fpcsim}.}
    \label{fig:q-beta-T}
\end{minipage}\hfill
\begin{minipage}{.45\textwidth}
    \hspace{-1.2cm}
    \includegraphics[width=1.3\textwidth]{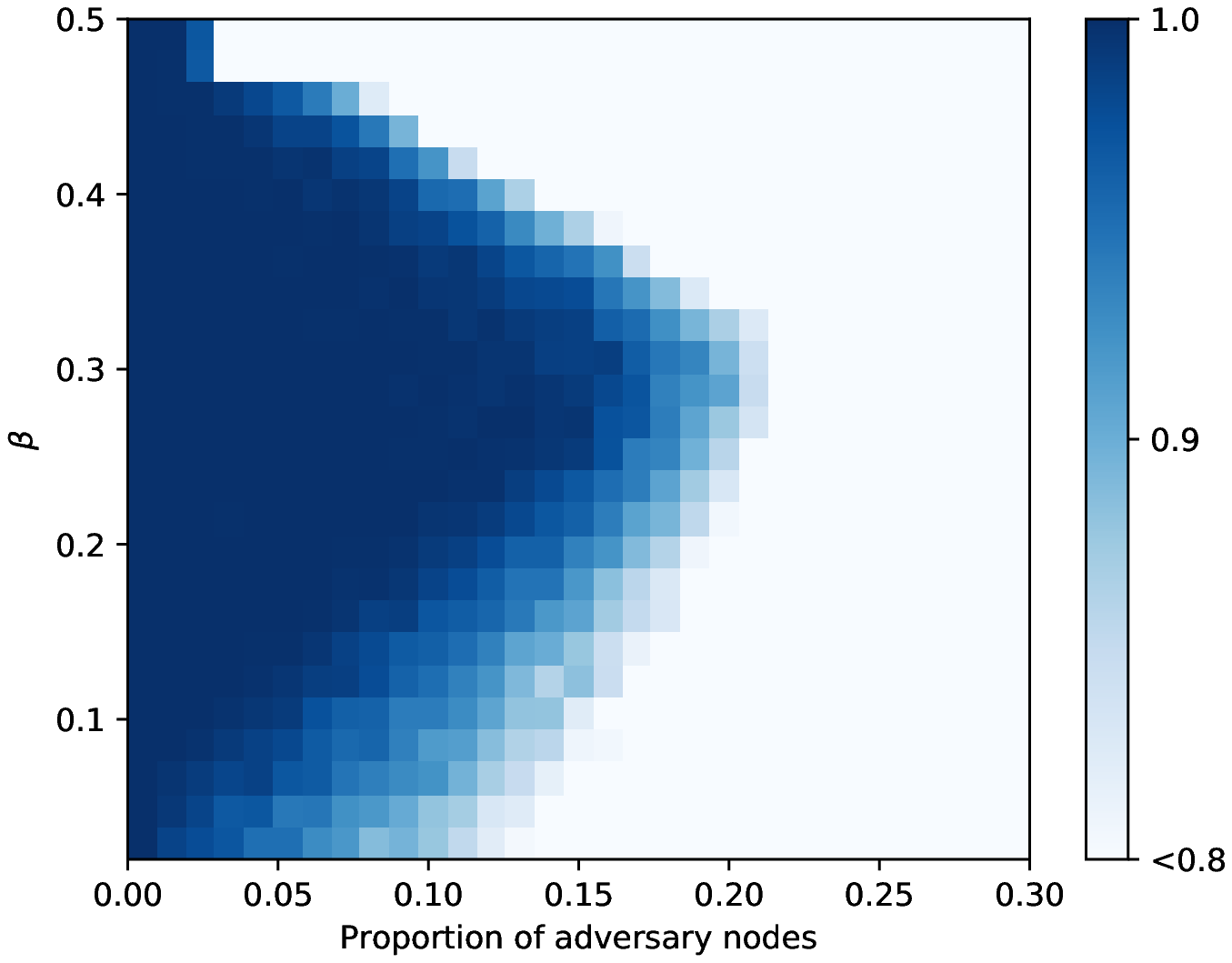}
    \vspace{-0.5cm}
    \caption{Agreement rate; \cite{fpcsim}.}
    \label{fig:q-beta-A}
\end{minipage}\hfill
\end{center}
\end{figure}

\subsection{Discussion}\label{sec:discussion}
In the remaining part of this section, we discuss which assumptions are necessary for FPC.

\paragraph{Random thresholds}
FPC relies in its construction on a (decentralized) random number generator that produces periodically random thresholds to break metastable situations. A first natural question is therefore how the robustness of FPC depends on the  robustness of the random numbers. In fact, the answer is that FPC is rather robust towards failures of the external randomness. To see this, it is good to remind oneself of the functionality of the random threshold: it serves the consensus protocol to escape meta-stable states. For this purpose, it is sufficient that from time to time some sufficiently large proportion agrees on a random threshold. Let us make this more precise. We need that, regardless of the past, 
with probability at least~$\theta$ (where $\theta>0$ is a fixed parameter)
the next outcome is a uniform random variable
which is unpredictable for the adversary. 
This random number is seen by at least 
$(1-\delta)$ proportion of honest nodes, 
where~$\delta$ is reasonably small. 
What we can prove in such a situation depends on what
the remaining $\delta(1-q)n$ honest nodes use as their decision thresholds: they can use some second candidate (in case there is an alternative 
source of common randomness), or they can choose their thresholds independently and randomly, etc.
Each of such situations would need to be treated separately, which is certainly doable but left out of this paper.
Let us note, though, that
the worst-case assumption 
is that the adversary can ``feed'' the (fake) decision thresholds to those $\delta(1-q)n$ honest
nodes. This 
would effectively mean that (at worst) these 
nodes would behave as cautious adversaries in the next round. This only matters if the random time~$\Psi$ did not yet occur.
Therefore, to obtain bounds 
like~\eqref{eq_safety_liveness_cautious}--\eqref{eq_safety_liveness_berserk} 
and~\eqref{eq_safety_liveness_semi} we 
can simply pretend that the value of~$q$ is increased by~$\delta$.

Assuming that $\delta=0$, it is easy to
figure out how this will affect our results: indeed,
in our proofs, all random thresholds matter only until~$\Psi$.
We can  obtain the following fact:
\begin{prop}
\label{p_weaker_RNG}
Assume the above on the random number generation
(with $\theta\in (0,1)$ and $\delta=0$). Then, 
the estimates 
\eqref{eq_safety_liveness_cautious}--\eqref{eq_safety_liveness_berserk} 
and~\eqref{eq_safety_liveness_semi}
remain valid (although, possibly, with other constants).
\end{prop}
Also, we want to stress out that 
partial control of the random numbers does not give 
access to a lot of influence (in the worst-case the adversary would delay the consensus a bit). Hence, there is not much need to be restrictive on the degree of decentralization for that part\footnote{In other words, it may make sense that different parts of the system are decentralized to a different degree.}: 
a smaller subcommittee can take care of the random numbers' generation, and some VDF-based random number generation scheme 
(e.g., such as those of~\cite{Lenstra_Wes17,wes}) 
may be used to further prevent this subcommittee from leaking the numbers before the due time.

We also refer to~\cite{fpcsim} for a numerical study of the FPC's dependence on the ``quality'' of the random thresholds.

\paragraph{Complete network view}
In the theoretical results of FPC, we assume that every node holds a complete network view. However, in permissionless systems or in less reliable networks with inevitable churns (nodes join and leave), this assumption is not necessarily verified. While churns can be modeled by faulty nodes, the situation of a partial view on the network participants deserves more attention. The property of whether two given nodes can know and query each other induces a natural network topology or underlying graph structure. This graph structure has a crucial impact on the performances of FPC (think about the extreme case where the diameter of the graph is large and the propagation of opinions may take a very long time). However, for situations where the graph is well connected, e.g., an expander graph, we believe that the theoretical results remain true and the proof strategy can be applied (however, with an increased technical difficulty). First numerical results in this direction were obtained in~\cite{manaFPC}, where several graph topologies are modelled.

\paragraph{Sybil protection}
The main security assumption on FPC is that an adversary is not supposed to control more than some given proportion (dependent on the threat model) of the total number of nodes. While in permissioned systems the total number of nodes can be controlled, the permissionless setting {\it a priori} allows an adversary to create an unlimited amount of identities; this is known as a so-called ``Sybil'' attack. Any real-world implementation of FPC in a permissionless system should therefore contain a Sybil protection. A natural generalization of FPC with a Sybil protection was given in \cite{fpcFairness, manaFPC} and some mathematical properties were studied further in \cite{fpcAsypFairness}. In this model, every node is given a certain weight (that may correspond to a scarce resource) or reputation. This weight is then used for choosing the queries and weighting the obtained opinions. We refer to \cite{manaFPC} for more details on the protocol and security assumption in terms of percentages of weights instead of just proportions of the nodes.

In view of the above results,
let us stress that one of the main features of FPC is that it turns  ``rather weak'' consensus (on the random numbers, the participants of the network, and their weights) into a ``strong'' consensus on the value of a bit (i.e., the validity of a transaction), with high probability. 


\section{Conclusion}\label{s:conclusion}
Majority dynamics are probably the most natural candidates for consensus protocols and promise scalable solutions for the DLT space. While protocols in this class behave well and their study is rather straightforward in an honest environment, the situation changes drastically in the presence of Byzantine actors. In particular, a possible attacker may keep the system in a so-called ``metastable state'' 
and prevent the honest nodes from finding consensus. 

We suggested analyzing these protocols as random walks on a potential. This method offers a heuristic to understand the typical behavior and 
main issues of the standard majority dynamics. Finally, we gave an overview of the FPC protocol that uses common random thresholds to prevent the occurrence of metastable states and allows formal proofs of the security of the protocol.

\section*{Acknowledgements}
The authors thank Olivia Saa for obtaining 
the solution of~\eqref{scary_integral_eq},
 Alexandre Reiffers-Masson and Yao-Hua Xu 
for suggesting several corrections to 
the arguments in Section~\ref{s_simple_model},
and David Phillips for helping us to improve the writing style.
We also thank the anonymous referees for their useful comments and suggestions.

\bibliographystyle{plain}
\bibliography{bibliography}

\end{document}